\documentclass{article}

\usepackage{booktabs} % For professional looking tables
\usepackage{multirow}
\usepackage{bbm}
\usepackage{xcolor}
\usepackage[normalem]{ulem}

\usepackage{enumitem}
\usepackage[utf8]{inputenc}
\usepackage{caption}
\usepackage{subcaption}

\usepackage{graphicx,epsfig,amsfonts,amsthm,amsmath,textcomp,listings,multirow}
\usepackage{adjustbox,lipsum,dsfont,xfrac,nicefrac,stmaryrd,mathtools,bm,amssymb,float}

\usepackage{yfonts}

\usepackage{tikz-cd}

\newcommand{\vertiii}[1]{{\left\vert\kern-0.25ex\left\vert\kern-0.25ex\left\vert #1 
    \right\vert\kern-0.25ex\right\vert\kern-0.25ex\right\vert}}
\newtheorem{theorem}{Theorem}[section]

\newtheorem{corollary}[theorem]{Corollary}
\newtheorem{lemma}[theorem]{Lemma}

\newtheorem{example}[theorem]{Example}
\newtheorem{proposition}[theorem]{Proposition}

\newtheorem{remark}[theorem]{Remark}
\theoremstyle{definition}
\newtheorem{definition}[theorem]{Definition}

\newtheoremstyle{case}{}{}{}{}{}{:}{ }{}
\theoremstyle{case}

\DeclareMathOperator*{\plim}{\text{$\mathbb{P}$-lim}}
\DeclareMathOperator*{\sym}{sym}
\DeclareMathOperator*{\osc}{osc}

\newcommand{\Ito}{\mathrm{It\hat{o}}}
\newcommand{\Strat}{\mathrm{Strat}}
\newcommand{\pvar}{p\text{-}\mathrm{var}}
\newcommand{\qvar}{q\text{-}\mathrm{var}}
\newcommand{\halfpvar}{\frac{p}{2}\text{-}\mathrm{var}}
\newcommand{\R}{{\mathbb R}}
\newcommand{\V}{{\mathcal{V}}}
\newcommand{\correction}[1]{{\textcolor{black}{#1}}}
\newenvironment{correctionenv}{\begin{color}{black}}{\end{color}}

\title{Gamma Hedging and Rough Paths}
\author{John Armstrong and Andrei Ionescu}

\begin{document}

\maketitle

\begin{abstract}
We apply rough-path theory to study the discrete-time gamma-hedging strategy. We show that if a
trader knows that the market price of a set of European options will be given by a diffusive pricing model,
then the discrete-time gamma-hedging strategy will enable them to replicate other European options so long
as the underlying pricing signal is sufficiently regular. This is a sure result and does not require that the
underlying pricing signal has a quadratic variation corresponding to a probabilisitic pricing model. We show how to generalise this result to exotic derivatives when the gamma is defined to be the Gubinelli derivative of the delta by deriving
rough-path versions of the Clark--Ocone formula which hold surely.

\correction{We illustrate our theory by proving that
if the stock price process is sufficiently regular, as is the implied volatility process of a European derivative with maturity $T$ and smooth payoff $f(S_T)$ satisfying $f^{\prime \prime}>0$, one can replicate with certainty any European derivative with smooth payoff and maturity $T$.}
\end{abstract}

\section{Introduction}

In the discrete-time gamma-hedging strategy, a trader seeks to replicate a derivative product by trading in $n$ hedging instruments, such as the underlying and other derivatives. They do this by choosing a pricing model and
then ensuring that at each time point the total delta and total gamma of their hedging portfolio match that of the derivative they seek to replicate.

We will prove that for many derivatives and pricing models, the 
profit and loss of this strategy converges to zero as the times
between hedges tends to zero on the assumption that the
pricing model accurately prices the hedging instruments given a pricing signal derived from market data. This is
a sure convergence result and holds so long as the pricing signal has sufficient regularity. There is no need for the pricing signal
to be compatible with any probabilistic model.

Let us give a precise statement in the case of European options. For simplicity
we will assume that there is a risk-free asset with deterministic interest rate $0$: this amounts to making an appropriate choice of numeraire. We will
assume throughout that $(\Omega,{\cal F},{\cal F}_t,{\mathbb Q})$ is a filtered probability space generated by a $d$-dimensional Brownian motion $W_t$ satisfying the usual conditions. 

\begin{theorem}
\label{thm:europeanOptionsTheorem}
Let $(S_t)_{t \in [0,T]}$ represent the realised path of \correction{a pricing signal}, 
$S_t \correction{ \in U } \subseteq {\mathbb R}^d$ \correction{for some open set $U$}. Suppose that there are $(n+1)$ options with
smooth bounded payoff functions $f^i(S_T)$ for
$i$ $(0 \leq i \leq n )$. Let the function $V^i(\correction{\hat{S},\hat{t}})$ 
\correction{for $\hat{S} \in U$ and $\hat{t} \in [0,T]$}
be defined to equal the risk-neutral price of the option
with payoff $f^i$ at time
$T$ given that the \correction{signal} follows the diffusion SDE
\begin{equation}
\label{eq:diffusionSDE}
d \tilde{S}_t = \correction{\mu(\tilde{S}_t,t) \, dt } + \sigma(\tilde{S}_t\correction{,t})\, d W_t.
\end{equation}
with initial condition $\tilde{S}_\correction{\hat{t}}=\correction{\hat{S}}.$
Here \correction{$\mu$ and $\sigma$} are a smooth function of linear growth in $\tilde{S}$ taking
values \correction{in $\R^d$} and the space of $d \times d$ matrices
\correction{respectively}. \correction{We assume that the $V^i$ are smooth. We assume that the $SDE$ ensures $\tilde{S}_t \in U$ almost surely.}

Suppose that a trader followed the gamma-hedging strategy
on a finite grid of time points $\pi \subset [0,T]$. This
means that they attempted
to replicate the option with payoff $f^0$ using the options with payoff $f^i$ for $1\leq i \leq n$ by holding $q^i_t$ units of option $i$ at each time 
point $t\in \pi$ where the $q^i_t$ are chosen
such that
\[
\sum_{i=1}^n q^i \frac{\partial V^i}{\partial S^\alpha} = \frac{\partial V^0}{\partial S^\alpha}, \quad
\sum_{i=1}^n q^i \frac{\partial^2 V^i}{\partial S^\beta \partial S^\alpha} = \frac{\partial^2 V^0}{\partial S^\beta \partial S^\alpha},
\]
for all indices $1\leq \alpha,\beta \leq d$ and placing any remaining wealth
in a risk-free account. If the initial wealth
of the trader is $V^0_0$, the $q^i$ are uniformly bounded and $S_t$
has finite $p$-variation for some $p<3$
then the final payoff will tend to $V^0_T$ as the mesh of the grid tends to zero.
\end{theorem}

\correction{Our requirement that the $V^i$ are smooth can be ensured, for example,
by requiring that $\sigma$ has bounded derivatives of all orders and satisfies
the H\"ormander condition: this requires that the vector space spanned repeated Lie brackets of its column vectors at each point is equal to $\R^d$ \cite[Theorem 2.3.3]{nualart2006malliavinCalculus}. If the signal were a typical stock price process, one would need to apply a logarithmic transformation to the SDE before applying a general-purpose theorem of this sort. 
}

\correction{The signal contains both information used to determine the payoff of derivative contracts and to predict market prices. We allow for a drift term in \eqref{eq:diffusionSDE} to allow
for the possibility that some components of the signal may not represent
traded assets. For example, the signal may contain current stock prices and a volatility parameter computed by calibrating a model to current option prices.
}

\medskip

\correction{Theorem \ref{thm:europeanOptionsTheorem}} can be interpreted as a theoretical model for the success of the market
practice of choosing a pricing model by calibrating to market options prices and then pursuing the gamma-hedging strategy without ever formulating a physical-measure model for the trajectories of the underlying. In classical delta hedging, one assumes the underlying follows a physical measure that is equivalent
to the pricing measure and obtains, at best, almost-sure results. Neither of these restrictions apply to gamma hedging. \correction{We illustrate
this in detail with the example of calibrating the Black--Scholes model in Section \ref{subsec:calibration}.}

\correction{If we assume some parts of the signal are
more regular, then it is not necessary to match
all the second-derivative terms.
We illustrate this by proving that
if the stock price process is sufficiently regular, as is the implied volatility process of a European derivative with maturity $T$ and smooth payoff $f(S_T)$ satisfying $f^{\prime \prime}>0$, one can replicate with certainty any European derivative with smooth payoff and maturity $T$ using just the stock and the derivative with payoff $f$. A precise statement is given in Theorem \ref{thm:blackScholesGamma}.}

\medskip

Although the statement of our theorem does not mention rough-path theory, it is the
basic tool that we use to prove our results. In rough-path theory, if one wishes to integrate some integrand against a signal, then the integrand must be similar to the signal in the
sense that it must possess a so-called Gubinelli derivative. The rough-path integral then consists of two terms: a classical term and a term involving the Gubinelli derivative. The classical
theory of delta hedging proceeds by using the \correction{m}artingale representation theorem to show
that the payoff of a derivative can be written as an integral of the delta against the changes in the underlying and then matching the delta of the hedging portfolio to that of the instrument to be replicated. When using rough-path theory, we must also match the terms arising from the Gubinelli derivatives. In the case of European options, we will see that the Gubinelli derivative
of the delta is equal to the gamma of the option, and we use this to define the gamma for path-dependent derivatives.

Rough-path theory has been used in model-free finance before. This paper builds on  \cite{armstrong2020option} which shows how an enhanced version of the classical delta hedge is required for a pathwise replication for European options, proposing a theoretical strategy of enhancing delta hedging with trades in a hypothetical form of volatility swap. 
The mathematics behind our approach has its origins in F\"ollmer's paper ``Calcul d'It\^o sans probabilit\'es'' \cite{follmer1981} which defines It\^o integration pathwise and proves a pathwise It\^o formula so long as the integrator admits a finite quadratic variation.
This can be used to obtain a pathwise theory of delta hedging subject to some restrictions on the grid of points used when discretising the delta-hedging strategy \cite{bick1994dynamic}.
Perkowski and Pr\"omel studied two model-free approaches to stochastic integration in \cite{perkowski2016}, one using the outer Vohk measure and the other rough-path integration, and Allan, Liu and Pr\"omel \cite{allan2024} extend this to include signals with jumps. Numerical results have been obtained in the model-free pricing and hedging of exotic options in \cite{lyons2020} using signature payoffs.

Our results are also related to Functional It\^o calculus \cite{cont2013functionalitocalculus,cont2016barcelona} which extends It\^o calculus to path-dependent functionals of stochastic processes using the derivatives of Dupire \cite{dupire2019functional}. The PhD thesis of Riga \cite{riga2015thesis} analyses continuous-time trading strategies pathwise, see also Ananova \cite{ananova2020} which relates rough differential equations with path-dependent coefficients to the weak derivative of Functional It\^o calculus. 
We do not consider Functional It\^o calculus in this paper, but note
that it is possible to relate the gamma as defined in this paper as a Gubinelli derivative with the natural definitions in terms of Dupire's derivatives: see \cite{andreiThesis}.

The topic of gamma hedging has been studied previously. In practice, delta hedging in discrete-time produces much larger hedging errors than gamma hedging does. Thus, it is common for traders to augment their strategies with gamma (or vega) terms to reduce this hedging error and reduce transaction costs. For the case of delta hedging smooth European payoffs, Heston and Zhou \cite{heston2000} show that the hedging error is inversely proportional to the number of re-balancing times, $N$. For gamma hedging, Gobet and Makhlouf \cite{gobet2012} give non-sharp lower bounds on the convergence rates for both equidistant and non-equidistant trading grids. They show a convergence rate of $N^{\frac{1}{2}}$ even for non-smooth payoffs. Brod\'{e}n and Wiktorsson \cite{broden2011} partially improve on this result by proving that gamma-hedging European option on an equidistant time grid in the Black--Scholes model has an order of convergence of $N^{\frac{3}{4}}$. On the efficient computation of the gamma of options, Fourni\'{e} et al \cite{fournie1999, fournie2001} use Malliavin calculus to compute the gamma of ``exotic European options'' in the the Black--Scholes model. Gobet and Kohatsu-Higa \cite{gobet2003} also use Malliavin calculus to compute the Greeks of barrier and lookback options. Dengler and Jarrow \cite{dengler1996} develop a discrete-time binomial option pricing model with random timesteps. They show that under this model, delta hedging approximates gamma hedging in the classical binomial model with fixed timesteps, which is the Black--Scholes model in its continuous-time limit.

Theorem \ref{thm:europeanOptionsTheorem} shows that the gamma-hedging strategy is extremely robust to a misspecification
of the probability model for the underlying, though it introduces a new dependency on the pricing model that determines how market prices for the hedging options react to changes in the underlying.
Our result can therefore be constrasted with the literature on the effect of hedging under a misspecified volatility. In \cite{elkaroui1998}, El Karoui, Jeanblanc-Picqu\'e and Shreve showed that hedging strategy computed under the assumption that the misspecified volatility dominates the true volatility provides a one-sided hedge for both European and American options. Analogous results are also shown when the true volatility dominates the misspecified volatility. Ahmad and Willmott \cite{ahmad2005} derive the formula for profits from delta hedging European options under a three-part structure: they explicitly differentiate between a true governing volatility of the underlying asset, inaccessible to the trader, an implied volatility characterised by option prices in the market and a hedge volatility reflecting the personal views of the trader. A thorough exposition of this result, known as the Fundamental Theorem of Derivatives Trading, can be found in the article by Ellersgaard, J\"onsson and Poulsen \cite{ellersgaard2017}.

More generally, there is an extensive literature on model-independent finance. Examples include Lyons \cite{lyons1995} which describes optimal and risk-free hedging strategies when the volatility is unknown, but is assumed to lie in some convex region depending on market prices. Avellanada et al \cite{avellaneda1995} describe pricing and hedging methods when the volatility is assumed to lie in some interval $(\sigma_{\min}, \sigma_{\max})$, see also Hobson \cite{hobson1998} and the subsequent development of the Skohorod Embedding Problem (SEP) to develop robust, model independent prices and hedges for exotic options described in \cite{hobson2011} and the survey article by Obł\'oj \cite{obloj2004}. It is related to the problem of \correction{M}artingale Optimal Transport (MOT) introduced by Beiglb\"ock, Henry-Labord\`ere and Penkner \cite{beiglbock2013}. In \cite{backhoffveraguas2019adapted}, hedges are shown to be robust under changes in the risk-neutral measure if one uses the adapted Wasserstein distance to measure the proximity of the model used for hedging to the model used to generate the path. The super-replication approach and the MOT problem are shown to be dual problems by Dolinsky and Soner \cite{dolinsky2014}.

The first contribution of the current paper is to consider the gamma-hedging strategy using the techniques of rough-path theory. In particular it shows that the trading strategy in volatility swaps required in the trading strategy of \cite{armstrong2020option}, can be replaced with a gamma-hedging strategy in conventional derivatives.

The second contribution of this paper is to show how Theorem \ref{thm:europeanOptionsTheorem} can be generalized to path-dependent derivatives. This requires identifying the Gubinelli derivative of the delta of an option (if it exists). Whereas \cite{armstrong2020option} applies PDE techniques to the analysis of European options, we show how to identify the gamma for large classes of derivatives using probabilistic techniques, and show how it can be related to the Malliavin derivative of the delta. We will show that barrier options do not in general admit such a gamma due to discontinuities of the delta when the barrier is hit, but we will also see how barrier options can be approximated by instruments with smoother payoffs which do have a gamma. Indeed we will show that all derivatives with a payoff that is continuous in the sup norm on the space of paths can be surely superhedged for a price arbitrarily close to the risk-neutral price. 

We remark that if one is content with probabilistic results when the
underlying is a semi-martingale then Theorem \ref{thm:europeanOptionsTheorem} can be proved using
classical stochastic calculus techniques, but we will only consider the rough-path approach in this paper. \correction{One reason for preferring
the rough-path approach is that applies even if the signal does
not arise from a probability model that is independent of the trader's actions: indeed the signal could be chosen by an adversary who already knows our strategy.}

Section \ref{sec:roughPathNotation} gives a brief outline of the notation we shall
use from rough-path theory. Section \ref{sec:gammaHedging} proves a general result showing the sure convergence of the gamma-hedging strategy subject to the assumption that rough-path integral formulae analogous to those in the \correction{m}artingale representation theorem hold. Section \ref{sec:european} applies this to European options, giving a proof of Theorem \ref{thm:europeanOptionsTheorem}. Section \ref{sec:pathDependent} extends the analysis to
path-dependent derivatives. Our proof strategy will be to use the Clark--Ocone formula to obtain almost-sure results allowing us to write the payoff as an integral and to identify the gamma. We will then use the relationship between the It\^o integral and the rough integral, together with the continuity of the rough integral, to obtain sure results. Section \ref{sec:pathDependentTheory} describes this general theory, Section \ref{sec:examples} applies this to an $L^2$ dense family of payoffs we call simple payoffs, to signature payoffs and to barrier options.  Section \ref{sec:superhedging} then demonstrates how one can superhedge derivatives surely.

\section{Rough-path theory notation}
\label{sec:roughPathNotation}

In this section we will describe the notation we use for rough-path theory. Unless stated otherwise, all definitions below are taken from \cite{frizhairer2020} and adjusted to appropriate $p$-variation regularity versions using \cite{friz2010}. We use the $p$-variation formulation as it makes our theory of controllable options easier to prove for one-touch options: the running maximum process is increasing so has finite $1$-variation and hence finite $p$-variation.

\subsection{Rough paths}

Given a metric space $(E,d)$, a function $f:[0,T] \to E$ and $p>1$, we will write $\lVert f\rVert_{p\text{-var};[0,T]}$ for the $p$-variation of $f$ (see e.g.\ \correction{\cite[p.\ 79]{friz2010}}), if it exists. We will write $C^{p\text{-var}}\left([0,T];E\right)$ for the set of all \correction{continuous} functions of finite $p$-variation from $[0,T]\rightarrow E$.

For a Banach space $\V$, we will define a $p$-variation norm
for functions of two variables $R_{(\cdot,\cdot)}:[0,T]\times[0,T]\rightarrow E$ by:
\begin{equation*}
    \lVert R\rVert_{p\text{-var};[0,T]} \coloneqq \left(\sup_{\pi\in\mathcal{P}[0,T]} \sum_{{[t_k,t_{k+1}]\in\pi}}\lvert R_{t_{k},t_{k+1}}\rvert^p\right)^\frac{1}{p}.
\end{equation*}
\correction{where $\mathcal{P}[0,T]$ is the set of partitions of $[0,T]$. We write $C^{\pvar}([0,T]^2;E)$ for the space of continuous functions with finite $p$-variation. We will write $\lVert \pi\rVert$ for the mesh size of the partition, that is the length of the largest sub-interval.}

Given a path $X_t$ in a linear space, we will write $X_{s,t}:=X_t-X_s$.

\begin{definition}[Rough path \cite{frizhairer2020}]
    For $p\in(2,3)$, we write $\mathcal{C}^{p\text{-var}}\left([0,T];\label{correction: start of higher dimensions}\V\right)$ for the \emph{space of $p$-variation rough paths} over $\V$ as those pairs $\mathbf{X}\coloneqq(X,\mathbb{X})$, where $X\in C^{p\text{-var}}([0,T];\V)$ and $\mathbb{X}\in C^{\frac{p}{2}\text{-var}}([0,T]^2;\V\otimes \V)$, such that
    \begin{enumerate}
        \item $\lVert X\rVert_{p\text{-var}}+\lVert \mathbb{X}\rVert_{\frac{p}{2}\text{-var}}<\infty$. 
        \item \emph{Chen's relation} holds: for all times $s,u,t\in[0,T]$, we have
        \begin{equation*}
            \mathbb{X}_{s,t}-\mathbb{X}_{s,u}-\mathbb{X}_{u,t} = X_{s,u} \otimes X_{u,t}.
        \end{equation*}
    \end{enumerate}
\end{definition}

If $\mathbf{X} = (X,\mathbb{X})\in \mathcal{C}^{p\text{-var}}([0,T];\V)$ is a rough path, we will call the $\V$-valued path $X$ its \emph{trace}.

The (homogeneous) norm of the rough path is defined to be:
\begin{equation*}
    \vertiii{\mathbf{X}}_{p\text{-var}} \coloneqq \lVert X\rVert_{p\text{-var}} + \sqrt{\lVert \mathbb{X}\rVert_{\frac{p}{2}\text{-var}}}.
\end{equation*}
For $\mathbf{X}=(X,\mathbb{X})$ and $\mathbf{Y}=(Y,\mathbb{Y})$ rough paths in $\mathcal{C}^{p\text{-var}}([0,T];\V)$, the (inhomogeneous) $p$-variation rough-path metric is defined to be
\begin{equation*}
    d_{p\text{-var}}(\mathbf{X},\mathbf{Y}) \coloneqq \lVert X-Y\rVert_{p\text{-var}} + \lVert \mathbb{X}-\mathbb{Y}\rVert_{\frac{p}{2}\text{-var}}.
\end{equation*}

We also recall a few technical definitions (see \correction{\cite[p.\ 22]{friz2010})}. \correction{Let $\Delta_T$ be the simplex $\{ (s,t) \subseteq [0,T]\times[0,T]: 0\leq s \leq t \leq T \}$.}
A map $w:\Delta_T \rightarrow [0,\infty)$ is said to be \emph{super additive} if for all $s\leq u\leq t$ in $[0,T]$,
\begin{equation*}
    w(s,u)+w(u,t)\leq w(s,t).
\end{equation*}
If in addition, $w$ is continuous, we say that $w$ is a \emph{control function} on $[0,T]$. \correction{
If $X \in C^{\pvar}([0,T];E)$, we write $w_{X,p}(s,t):=\|X\|_{\pvar;[s,t]}^p$ to be the control function
defined by the given power of $p$-variation over
the interval $[s,t]$. This is called the $p$-variation control. We define the $p$-variation control
$W_{R,p}:=\|R\|_{\pvar;[s,t]}^p$ for elements of $R\in C^{\pvar}([0,T]^2;E)$ similarly.
}

\begin{correctionenv}
\begin{lemma}
The sum of two control functions is a control function.
If $w_1$ and $w_2$ are control functions and if $\alpha, \beta>0$ with $\frac{1}{\alpha}+\frac{1}{\beta}>1$
then $w_1^\alpha w_2^\beta$ is a control.
\label{lemma:productControls}
\end{lemma}
This is proved in \cite[Exercise 1.10]{friz2010}.
\end{correctionenv}

Given a finite partition $\pi \subset [0,T]$ and a control $w$, the \emph{modulus of continuity} of $w$ on a scale smaller or equal than the mesh size $\lVert \pi\rVert$ is defined as
\begin{equation*}
    \osc(w, \lVert \pi\rVert) \coloneqq \sup\{ w(s,t): \lvert t-s\rvert \leq \lVert \pi\rVert\}.
\end{equation*}
\begin{correctionenv}
Control functions are useful in integration theory as for any $\epsilon>0$ we have
that
\begin{align}
\sum_{[s,t]\in \pi} w(s,t)^{1+\epsilon} &\leq 
\left(\sup_{[s,t]\in \pi} w(s,t)^\epsilon\right) \left(\sum_{[s,t]\in \pi} w(s,t) \right) \nonumber \\
&\leq 
\osc(w^\epsilon, \lVert \pi\rVert)\, w(0,T) \nonumber \\
&\to 0 \quad \text{as} \quad \|\pi\|\to 0.
\label{eqn:controlLimit}
\end{align}
\end{correctionenv}

\subsection{Controlled rough paths}

Let $\hat{\V}$ denote a Banach space, possibly different from $\V$.

 A pair $(Y,Y')\in C^{p\text{-var}}\left([0,T];\hat{\V}\right)\times C^{{\correction{p}}\text{-var}}\left([0,T];\mathcal{L}(\V;\hat{\V})\right)$ is said to be a $X$\emph{-controlled path of $\correction{p}$-variation regularity} if $R^Y_{s,t}\coloneqq Y_{s,t}-Y'_sX_{s,t}$ \correction{satisfies $R^Y \in C^{\frac{p}{2}}([0,T]^2, {\cal V})$}. We write $\mathcal{D}_X^{\correction{p}\text{-var}}\left([0,T];\hat{\V}\right)$ for the set of all $X$-controlled paths of $\correction{p}$-variation regularity taking values in $\hat{\V}$. If the domain of $Y$ is clear from the context, we will abbreviate this to $\mathcal{D}_X^{\correction{p}\text{-var}}$.

If $(Y,Y^\prime) \in \mathcal{D}_X^{\correction{p}\text{-var}}$ we will say that $Y^\prime$ is a \emph{Gubinelli derivative} of $Y$ with respect to $X$.

$\mathcal{D}_X^{\correction{p}\text{-var}}$ is a Banach space with norm
\begin{equation*}
    \lVert (Y,Y')\rVert_{\mathcal{D}_X^{\correction{p}\text{-var}}} \coloneqq \lvert Y_0\rvert + \lvert Y'_0\rvert + \underbrace{\lVert Y'\rVert_{p\text{-var}} + \lVert R^Y\rVert_{\correction{\frac{p}{2}}\text{-var}}.}_{\coloneqq \lVert Y,Y'\rVert_{X,\correction{p}\text{-var}}}
\end{equation*}

When the driving noise $X$ is random path of a diffusion then the Gubinelli derivative will almost surely be unique (see the discussion of truly-rough paths in \correction{\cite[Chapter 6]{frizhairer2020})}. Therefore in our applications to finance we will often be able to talk unambiguously of the Gubinelli derivative.

\subsection{The rough integral}

In rough-path theory, rough paths play the role of integrators and 
controlled rough paths play the role of integrands. Together they
allow us to define the rough integral.

\begin{theorem}[Gubinelli \cite{gubinelli2004}] \label{thm: gubinelli theorem}
Given a controlled path
\[
(Y,Y') \in \mathcal{D}_X^{\correction{p}\text{-var}}\left([0,T];\mathcal{L}(\V,\hat{\V})\right),
\]
the limit
\begin{equation} \label{eq: definition of rough integral}
    \int_{0}^{T}Yd\mathbf{X} \coloneqq \lim\limits_{\lVert \pi\rVert\rightarrow 0} \sum_{[s,t]\in\pi} \left(Y_sX_{s,t}+Y'_s\mathbb{X}_{s,t}\right)
\end{equation}
exists and satisfies the estimate
\begin{equation} \label{eq: bound on rough integral}
    \left\lvert \int_s^t Yd\mathbf{X} - Y_sX_{s,t} - Y'_s\mathbb{X}_{s,t} \right\rvert
    \leq \correction{C(p)} \omega(s,t)^{\correction{\frac{3}{p}}},
\end{equation}
where the constant $C(\correction{p})$ is global and only depends on $p$,
\correction{and $w(s,t)$ is the control function}
\begin{equation*}
\correction{w(s,t)=}
    \lVert R^Y\rVert^\correction{\frac{p}{3}}_{\correction{\halfpvar;[s,t]}}\lVert X\rVert^\correction{\frac{p}{3}}_{\pvar;[s,t]}+\lVert Y'\rVert^\correction{\frac{p}{3}}_{\pvar;[s,t]}\lVert \mathbb{X}\rVert^\correction{\frac{p}{3}}_{\halfpvar;[s,t]}.
\end{equation*}
Moreover, the map
\begin{equation*}
    (Y,Y')\mapsto (Z,Z')\coloneqq \left(\int_0^{\cdot} Yd\mathbf{X},Y\right)
\end{equation*}
from $\mathcal{D}^{\correction{p}\text{-var}}_X\left([0,T];\mathcal{L}(\V,\hat{\V})\right)$ to $\mathcal{D}^{\correction{p}\text{-var}}_X([0,T];\hat{\V})$ is continuous.
\end{theorem}
When we wish to highlight the role of the Gubinelli derivative we will use the notation $\left((Y,Y')\cdot(X,\mathbb{X})\right):=\int Yd\mathbf{X}$.

The rough integral allows one to define rough differential equations, and
prove existence of uniqueness of solutions subject to appropriate assumptions
on the coefficients, see \cite{frizhairer2020}. The solution map for rough differential equations is continuous in the driving noise in the rough-path
topology.

\section{Gamma Hedging}
\label{sec:gammaHedging}

In this section we will state and prove our main result, which is to define the gamma-hedging strategy and to estimate its error. For simplicity, we will consider a financial market with risk-free rate $0$.

\begin{theorem}
\label{thm:gammaHedging}
Let $(S,{\mathbb S}) \in \mathcal{C}^{p\text{-var}}\left([0,T],\V\right)$ be a rough path representing a pricing signal.

Assume that the market contains $n$ financial instruments
that we will use for hedging and that the price $V^i_t$ of instrument $i$ ($1\leq i \leq n$) at each time $t$ satisfies the rough-integral formula
\begin{equation}
\label{eq:representableAsIntegral}
V^i_t = V^i_0 - \correction{\int_0^t \Delta^i_u \cdot m_u \, du} + \left((\Delta^i_u, \Gamma^i_u) \cdot (S,{\mathbb S}) \right)_{0,t}
\end{equation}
where \correction{$m \in C^0([0,T];  \V)$, the first integral is a Riemann integral}, and $(\Delta^i, \Gamma^i) \in \mathcal{D}_X^{\correction{p}\text{-var}}\left([0,T];{\cal L}(\V,\mathbb{R})\right).$

Suppose that $V^0_t$ also satisfies equation \eqref{eq:representableAsIntegral} and that we can find bounded quantities
$q^i_t$ such that for all $t\in[0,T]$
\begin{equation} \label{eq: gamma hedging condition}
    \begin{cases}
        \sum_{i=1}^n q^i_t \Delta^{i}_t &= \Delta^{0}_t, \\
        \sum_{i=1}^n q^i_t \Gamma^{i}_t &= \Gamma^{0}_t.
     \end{cases}
\end{equation}

Let $\pi \subset [0,T]$ be a grid of time points. Suppose
that a trader has initial wealth $V^0_0$ and that they
ensure that they holds $q^i_t$ units
of each hedging instrument $i$ at each time $t \in \pi$, and placing the rest of
their wealth in the risk-free account. This is the discrete-time gamma-hedging strategy. If $\Pi^\correction{\pi}_t$ denotes the wealth of the
trader at each time then $|\Pi^{\correction{\pi}}_T - V^0_T| \to 0$ as $\|\pi\|\to 0$.
\end{theorem}

\begin{proof}

We will allow $C$ to vary line by line our proof in the usual way.
\correction{We will also write $q^0=-1$.}
Let $[s,t]\in\pi$ be an interval of our discrete trading grid. The increment of our wealth from $s$ to $t$, call it $\Pi^\pi_{s,t}$ satisfies
	\begin{equation*}
	\Pi^\pi_{s,t} = \sum_{i=1}^n q^i_sV^{i}_{s,t} 
	\end{equation*}
	Hence
    \begin{align}
        \lvert \Pi^{\pi}_T-\correction{V^0_T}\rvert &= \left\lvert \Pi^\pi_{0,T}-(V^0_T-V^0_0)\right\rvert \nonumber \\ 
        &= \left\lvert \sum_{[s,t]\in\pi}\Pi^\pi_{s,t} + \correction{\int_0^T \Delta^{0} \cdot m\, d u} - \int_0^T \Delta^{0}\cdot d\mathbf{S}
        \right\rvert \nonumber \\
        &= \left\lvert \sum_{[s,t]\in\pi} \left( \sum_{i=1}^n q^i_s V^{i}_{s,t} + \correction{\int_s^t \Delta^{0} \cdot m \, du} - \int_s^t \Delta^{0} \cdot d\mathbf{S} \right) \right\rvert \nonumber \\
        &= \Bigg\lvert \correction{\sum_{[s,t]\in\pi}} \sum_{i=0}^{n} q^i_s \left(
        \correction{-\int_s^t \Delta^{i} \cdot m \, du} + 
        \int_s^t \Delta^{i} \cdot d\mathbf{S}
        \right)
        \Bigg\rvert, \label{eq:errorEstimate}
    \end{align}
    where we have used equation \eqref{eq:representableAsIntegral} in the last line. \correction{We first bound the rough-integral term in this expression.} Using (\ref{eq: gamma hedging condition}) and the triangle inequality:
    \begin{align*}
        \left\lvert \sum_{i=0}^{n} q^i_s \int_s^t \Delta^{i}\cdot d\mathbf{S} \right\rvert 
        &= \left\lvert \sum_{i=0}^{n} q^i_s \int_s^t \Delta^{i}\cdot d\mathbf{S} - \sum_{i=0}^{n} q^i_s\left(\left(\Delta^{i}_s\right){S}_{s,t}+\Gamma^{i}_s\mathbb{S}_{s,t}\right) \right\rvert \\
        &\leq \sum_{i=0}^{n}\lvert q^i_s \rvert\left\lvert \int_s^t \Delta^{i} \cdot d\mathbf{S} - \left(\left(\Delta^{i}_s\right){S}_{s,t}+\Gamma^{i}_s{\correction{\mathbb{S}}}_{s,t}\right) \right\rvert \\
        &\leq C\left(\max_{0\leq i\leq {n}}\lVert q^i\rVert_{\infty} \right) \sum_{i=0}^n \correction{w_i(s,t)^{\frac{3}{p}}} 
    \end{align*}
    where $w_i$ denotes the control function defined in Gubinelli's theorem for the integral of $\Delta^i$.
    \begin{correctionenv}    
    Since $p<3$ and $w_i$ is a control function,
    we may use the estimate
    \eqref{eqn:controlLimit}
    to see that the size of the rough-integral terms in \eqref{eq:errorEstimate} tends to 0 as $\|\pi\|\to 0$.
    
    We now bound the Riemann integral in equation \eqref{eq:errorEstimate} in a similar fashion. Using (\ref{eq: gamma hedging condition}) as before we find
    \begin{align*}
        \left\lvert \sum_{i=0}^{n} q^i_s \int_s^t \Delta^{i}\cdot m \, du \right\rvert    
&\leq \left\lvert \sum_{i=0}^{n} q^i_s  \int_s^t \Delta^{i} \cdot m \, du
-
\sum_{i=0}^{n} q^i_s \left(\Delta^{i}_s \cdot m_s \right)(t-s) \right\rvert  \\             
&\leq \sum_{i=0}^{n}\lvert q^i_s \rvert\left\lvert \int_s^t \Delta^{i} \cdot m \, du - \left(\Delta^{i}_s \cdot m_s \right)(t-s) \right\rvert         \\
&\leq C \sum_{i=0}^{n}\lvert q^i_s \rvert\, \osc( \Delta^i \cdot m , \pi ) |t-s|.
    \end{align*}
    We have used bounds on the Riemann integral to obtain the final line. Hence
    \begin{align*}
    \left\lvert \sum_{[s,t]\in\pi} \sum_{i=0}^{n} q^i_s \int_s^t \Delta^{i}\cdot m \, du \right\rvert
    &\leq  C T \max_{0\leq i\leq {n}}\lVert q^i\rVert_{\infty}\, \osc( \Delta^i \cdot m , \pi )    \\
    &\to 0 \text{ as } \|\pi\|\to 0.
    \end{align*}    
This completes the proof.
\end{correctionenv}
\end{proof}

\begin{correctionenv}
\begin{remark}
Note that we can weaken the second equation in \eqref{eq: gamma hedging condition}
by only requiring that there exists a lift ${\mathbb S}_{s,t}$ such that
\begin{equation*}
    \sum_{i=1}^n q^i_s \Gamma^{i}_s \, {\mathbb S}_{s,t}= \Gamma^{0}_s \, {\mathbb S}_{s,t} \qquad \forall\, 0<s<t<T.
\end{equation*}
This may make it possible to reduce the number of hedging instruments required.
\end{remark}
\end{correctionenv}

According to the terminology of Theorem \ref{thm:gammaHedging}, the gamma-hedging strategy is only defined when equation \eqref{eq:representableAsIntegral} holds for all the hedging assets, and this equation can be taken as defining the meaning of $\Delta$ and $\Gamma$. We will see how these can be related to the familiar Greeks of European options in Section \ref{sec:european}. This will justify our notation and choice of terminology. We will consider path-dependent derivatives in Section \ref{sec:pathDependent}.

In Theorem \ref{thm:gammaHedging} the signal $(S, {\mathbb S})$ need not be traded. We simply need that such a rough path exists. However, one specific application we have in mind is when some component of the trace $S$ represents the path of an underlying and the $V^i$ represent the price of derivatives on this underlying. If $e \in {\cal L}(\V, {\mathbb R})$ is a coordinate functional identifying such a component, then ${\mathbf e}:=(e,0)$ is an $S$-controlled rough path and satisfies equation \eqref{eq:representableAsIntegral} if we take $\Delta=e$,  $\Gamma=0$\correction{, and $m \in {\ker e}$.}

\section{European options}
\label{sec:european}

In this section we will consider the case when all the assets of
Theorem \ref{thm:gammaHedging} are intended to represent European derivatives
on $S$ with payoff $f^i(S_T)$. In this case we will be able to find conditions when equation \eqref{eq:representableAsIntegral} holds using the results of
\cite{armstrong2020option} which we now summarize.

Let $\sigma:\mathbb{R}^{d}\rightarrow \mathbb{R}^{d\times d}$ be a smooth function.
Suppose that the prices of the hedging instruments at time $t$ are given by smooth functions $V^i(S_t,t)$ satisfying the \correction{Feynman--Kac} PDE:
\begin{equation}
\frac{\partial V^i}{\partial t} + \correction{\sum_{\alpha=1}^d \mu_\alpha \frac{\partial V^i}{\partial S^\alpha}} + \frac{1}{2}\sum_{1\leq \correction{\alpha,\beta}\leq d} (\sigma \sigma^\top)_{\correction{\alpha,\beta}}\frac{\partial^2 V^i}{\partial S^\correction{\alpha} \partial S^\correction{\beta}} = 0
\label{eq:blackScholesPDE}
\end{equation}
with boundary condition $V^i(S_T,T)=f^i(S_T).$ Suppose also that $(S,{\mathbb S})$ is a rough path. We can combine this with the smooth path $t$ with enhancer $0$ to obtain a rough path $(X, {\mathbb X})$. Letting $D$ denote the gradient operator, $(DV(X), D^2 V(X)$ will then be an $X$-controlled rough path.
Applying the rough-path It\^o's Lemma \cite[Proposition 5.8]{frizhairer2020} to  ${\mathbf X}$ we find:
\begin{equation}
\label{eq:roughIto}
V^i(S_t,t)=V^i(S_0) + \int_0^t D V^i \cdot d S
+ \int_0^t \frac{\partial V^\correction{i}}{\partial t} \cdot d \correction{u}
+ \int_0^t \correction{\frac{1}{2}} D^2 V^i d [\mathbf S]_\correction{u}.
\end{equation}
Here $D$ denotes the gradient operator and the
{\em rough bracket} $[\mathbf S]_t$ is defined by
\[
[\mathbf{S}]_t = S_{0,t} \otimes S_{0,t}-\correction{2}\sym({\mathbb S}_{0,t}),
\]
where $\sym$ denotes tensor symmetrization. In equation \correction{\eqref{eq:roughIto}} the first integral is a rough-path integral, but the last two integrals are Young integrals. 

If we can assume that
\begin{equation}
\label{eq:bracketCondition}
[\mathbf{S}]_t = \int_0^t \sigma(S_u, \correction{u}) \sigma(S_u, \correction{u})^\top \, du
\end{equation}
equations \eqref{eq:blackScholesPDE} and \correction{\eqref{eq:roughIto}} will together imply
equation \eqref{eq:representableAsIntegral} when \correction{$m_t=\mu(S_t,t)$}, $\Delta^{\correction{i}}=D V^i$ and $\Gamma^{\correction{i}}=D^2 V^i$.

It is not obvious that, given a path $S$ of finite $p$-variation, we can
then find a
lift ${\mathbb S}$ such that equation \eqref{eq:bracketCondition} holds.
However, we note that in the rough-path \correction{It\^o's} lemma, equation \eqref{eq:roughIto},
the integrand in the \correction{final} integral is the symmetric tensor $D^2 V$.
As a result, we do not need to identify a full lift for $S$, only for the symmetric part \correction{of the lift}. This motivates the definition of reduced rough paths.

\begin{definition}[Reduced rough paths \cite{frizhairer2020}]
$\mathbf{X}=(X,\mathbb{X})\in\mathcal{C}^{p\text{-var}}([0,T];\V)$ is said to be a \emph{reduced rough path} if the lift $\mathbb{X}$ takes values in $\text{Sym}(\V\otimes \V)$ and satisfies the reduced Chen relation
\begin{equation*}
    \mathbb{X}_{s,t}-\mathbb{X}_{s,u}-\mathbb{X}_{u,t} = \text{Sym}\left(X_{s,u}\otimes X_{u,t}\right),\forall s,u,t\in[0,T].
\end{equation*}
\end{definition}

So long as one is willing to restrict attention  to integrands with symmetric
Gubinelli derivatives, one can then define a theory of rough integration for
reduced rough paths and \correction{It\^o's} lemma, equation \eqref{eq:roughIto}, will still
hold, as will a version of Theorem \eqref{thm:gammaHedging}.

The advantage of reduced rough paths is that their existence theory
is trivial: given $S \in C^{p-\text{var}}([0,T];\V)$ and
$\gamma \in C^{\frac{p}{2}-\text{var}}([0,T]; \text{Sym}(\V \otimes \V))$ we may
define a reduced rough-path with bracket $\gamma_t$ by
\begin{equation}
{\mathbb S}_{s,t}=\frac{1}{2}(S_{s,t}\otimes S_{s,t} + \gamma_{s,t}).
\label{eq:enhancerFromBracket}
\end{equation}
In particular, for any $S$ of finite $p$-variation for $p \in (2,3)$
we can find a reduced rough path $(S, {\mathbb S})$ satisfying equation
\eqref{eq:bracketCondition}, which then implies that equation
\eqref{eq:representableAsIntegral} holds.

Thus if the market prices for the hedging instruments all satisfy
\eqref{eq:blackScholesPDE}, then we can replicate the payoff $f^0_T$ by charging $V^0_0$ and following the gamma-hedging strategy, where $V^0_t$ is
the solution to \eqref{eq:blackScholesPDE} with the appropriate terminal condition.

Our treatment so far has not mentioned probability theory, but we
now note that the prices $V^i_t$ have a clear probabilistic interpretation.

Let
$\correction{\tilde{S}}_t$ follow a
$d$-dimensional It\^o diffusion given by \eqref{eq:diffusionSDE}. Then by the Feynman--Kac theorem, if
\[
V^i(S,t):=E_t(f^i(\correction{\tilde{S}}_T) \, \mid \, \correction{\tilde{S}}_t=S)
\]
then $V^i$ will satisfy \eqref{eq:blackScholesPDE}.

Combining these observations we have proved Theorem \ref{thm:europeanOptionsTheorem} from the introduction.

\subsection{Calibration to the Black--Scholes model}
\label{subsec:calibration}

\begin{correctionenv}
As an application of our results we will consider the options market as it existed
before the stock market crash of October 1987. Before the crash, there was no
pronounced volatility smile, but it has subsequently become a persistent feature of the market \cite{rubinstein1985nonparametric,rubinstein1994implied}.

Let us suppose that we have $n$ call options on a stock, $S$. All
the options have maturity $T$, and they have strikes $K^i$ ($i\leq 1 \leq n$).
We will suppose that at each time there
is a single implied volatility $\sigma_t$ such that the price of each option is given by 
\begin{equation}
V^i_t = \text{BS}(S_t, K^i, \sigma_t, T-t)
\label{eq:blackScholesPricingModel}
\end{equation}
where $\text{BS}$ is the Black--Scholes pricing formula for call options with $r=0$.
This is the simplest form of calibration of an option pricing model, as we use only one calibration parameter $\sigma$.

We will consider the pair $(S_t,\sigma_t)$ to be our signal. We note that the price
$\text{BS}(\hat{S}, K^i, \hat{\sigma}, T-\hat{t})$
of our call options is equal to the risk-neutral prices in the diffusion model
\begin{equation}
\begin{split}
d \tilde{S}_t &= \tilde{\sigma}_t \, d W^1_t \\
d\tilde{\sigma}_t&= 0
\label{eq:blackScholesModel}
\end{split}
\end{equation}
with initial condition $(\tilde{S}_{\hat{t}},\tilde{\sigma}_{\hat{t}})=(\hat{S},\hat{\sigma})$. Since the
Black--Scholes formula is smooth for all $t<T$,
Theorem \ref{thm:europeanOptionsTheorem} applies.

If $V^i$ represent the prices of any European derivatives in the Black--Scholes model then the
maximal rank of the matrix
\begin{equation}
\left(\frac{\partial V^i}{\partial S},
\frac{\partial V^i}{\partial \sigma},
\frac{\partial^2 V^i}{\partial S^2},
\frac{\partial^2 V^i}{\partial S \partial \sigma},
\frac{\partial^2 V^i}{\partial \sigma^2}\right)_{1\leq i\leq n}
\label{eq:greeksMatrix}
\end{equation}
is $4$. To see this, let $q_{S_0,\sigma,\tau}(u)$ denote the probability density at time $\tau$ for a geometric Brownian motion with drift 0, constant volatility $\sigma$ and initial condition $S_0$. The Black--Scholes price at time $t$ of a derivative with payoff $f(S_T)$ is
\begin{equation}
\mathrm{BS}^f(S_t, \sigma_t, \tau) :=
\int_0^\infty f(u) q_{S_t,\sigma,\tau}(u) \, du.
\label{eq:priceAsIntegral}
\end{equation}
where $\tau=T-t$ is the time to maturity.
The scaling behaviour of geometric Brownian motion under time dilations ensures that $q(S_0,\frac{1}{\alpha} \sigma, \alpha^2 \tau)=q(S_0,\sigma,\tau)$ for all $\alpha>0$. Taking partial derivatives of this equation with respect to $\alpha$ then setting $\alpha=1$ yields the identity
\[
2 \tau \frac{\partial q}{\partial \tau} - \sigma \frac{\partial q}{\partial \sigma} = 0.
\]
The Fokker--Planck equation for geometric Brownian motion gives
\[
\frac{\partial q}{\partial \tau}=\frac{1}{2}S^2 \sigma^2 \frac{\partial^2 q}{\partial S^2}.
\]
We deduce that
\[
\sigma \tau S^2 \frac{\partial^2 q}{\partial S^2} = \frac{\partial q}{\partial \sigma}.
\]
Hence we obtain from the integral formula \eqref{eq:priceAsIntegral} that
\begin{equation}
\sigma \tau S^2 \frac{\partial^2 V^i}{\partial S^2} = \frac{\partial V^i}{\partial \sigma}.
\label{eq:vegaProportionalToGamma}
\end{equation}

Therefore, given $t<T$, assuming the signal $(S,\sigma)$ has finite $p$-variation for $p<3$ and
that the matrix \eqref{eq:greeksMatrix} has rank 4 for all times up to and including $t$, we can
replicate the derivative payoff 
$BS^f(S_t,\sigma_t,T-t)$
on our signal, by charging the Black--Scholes price for the European option with payoff $f$
and maturity $T$ at time 0 and then pursuing the full gamma-hedging strategy of matching all the first and second order derivatives of the price with respect to $S$ and $\sigma$.

We call this the full gamma-hedging
strategy to contrast with the classical gamma-hedging strategy where only the partial derivatives with respect to $S$ are matched. If we are willing to assume that $\sigma$ is sufficiently smooth, the classical gamma-hedging strategy will suffice, as the next result shows.

\begin{theorem}
\label{thm:blackScholesGamma}
Suppose that a European option on an underlying $S$ with maturity $T$ and convex, non-linear, payoff $f_1(S_T)$ is traded at all times for the price $\mathrm{BS}^{f_1}(S_t,\sigma_t,T-t)$ where $S_t \in C^{\pvar}([0,T],{\mathbb R}^{+})$, $\sigma_t \in C^{\qvar}([0,T],{\mathbb R}^{+})$
with $p<3$, $q<2$ and $\frac{1}{p}+\frac{1}{q}>1$. Given a smooth function $f_0$ and $t \in [0,T)$,  the payoff $BS^{f_0}(S_t,\sigma_t,T-t)$
can be replicated at the Black--Scholes price in the sense that the profit or loss of the discrete-time classical gamma-hedging strategy tends to zero as the mesh tends to zero. Here the classical gamma-hedging strategy is to be understood as matching only the derivatives $\frac{\partial V}{\partial S}$ and $\frac{\partial^2 V}{\partial S^2}$ at each time.
If $f_1^{\prime\prime}(S_T)>0$ for all $S_T>0$ then it is possible
to replicate the payoff $f_0$ at maturity.
\end{theorem}
\begin{proof}
We will label the option with payoff $f_1$ as instrument $1$ and the stock as instrument $2$ in order to apply the notation from the rest of the paper.

The purpose of our convexity assumption is to ensure that the Black--Scholes gamma of the option $f_1$ is positive at times $t<T$, see Lemma \ref{lemma:convexity} in the appendix for a proof.
Because the gamma of the stock is always zero and its delta is always one, the matrix
\[
\left(
\begin{array}{cc}
\frac{\partial V^1}{\partial S} & \frac{\partial^2 V^1}{\partial S^2} \\
\frac{\partial V^2}{\partial S} & \frac{\partial^2 V^2}{\partial S^2}
\end{array}
\right)
\]
has rank $2$. Hence we we can find bounded quantities $q^1_t$, $q^2_t$ such that
\[
\sum_{i=0}^2 q^i_t \frac{\partial V}{\partial S}^i=0, \qquad \sum_{i=0}^2 q^i_t \frac{\partial^2 V}{\partial S^2}^i=0,
\]
with $q^0_t=-1$ for all times $t\in[0,T)$. With the additional assumption that $f_1^{\prime\prime}(S_T)>0$ for all $S_T>0$, we can find such $q^i_t$ for all $t \in [0,T]$.

Our first aim will be to write an integral formula
for the price using a rough-path integral in $S_t$ and a Young integral in $\sigma_t$.
We achieve this in equation \eqref{eqn:roughIntegral2dB} below. However,
to understand this equation we need to
introduce
a slightly more flexible definition of a controlled rough path to ensure that the
integrand is controlled.

Let $(X,{\mathbb X})$ be a $p$-rough path. Let $r$ be chosen such that $\frac{1}{r} + \frac{2}{p}>1$. Suppose that $(Y,Y^\prime) \in C^{\pvar}([0,T];\hat{\V}) \times C^{r\text{-var}}([0,T];\mathcal{L}(\V;\hat{\V}))$ and that 
$R^Y_{s,t}\coloneqq Y_{s,t}-Y'_sX_{s,t}$ has finite $\left(\frac{1}{p}+\frac{1}{r}\right)^{-1}$-variation
then we will say that $(Y,Y^\prime)$ is $X$-controlled with $(p,r)$-variation regularity. The case $r=p$ corresponds to the notion of controlled
rough path we have used previously. We can define a rough integral for such paths and Gubinelli's theorem still holds with an appropriate choice for the control function \cite{armstrong2020option}.

For our application we want to view the price
process as a controlled rough-path with Gubinelli derivative given by the delta. We
know from the theory when $r=p$ that the remainder
\[
R^{V^i}_{s,t}:=V^{i}_{s,t}-\frac{\partial V^i}{\partial S}\Big|_s S_{s,t} -\frac{\partial V^i}{\partial \sigma}\Big|_s \sigma_{s,t}
\]
has finite $\frac{p}{2}$-variation regularity. The third term has finite $q$-variation regularity. Hence 
$(V^i, \frac{\partial V^i}{\partial S})$ will be an $S$-controlled rough
path of $(p,r)$-variation regularity so long as:
\begin{enumerate}[label=(\roman*)]
\item $\frac{1}{r}+\frac{2}{p}>1$ (from the definition of $(p,r)$-variation regularity);
\item $\frac{1}{p} \geq \frac{1}{r}$ (to ensure $\frac{\partial V^i}{\partial S} \in C^{r\text{-var}}$);
\item $\frac{1}{q} \geq \frac{1}{p}+\frac{1}{r}$ and $\frac{2}{p} \geq \frac{1}{p}+\frac{1}{r}$ (to ensure $R$ is sufficiently regular).
\end{enumerate}
These conditions are equivalent to $\frac{1}{p}\geq \frac{1}{r}>1-\frac{2}{p}$ and $\frac{1}{q}-\frac{1}{p} \geq \frac{1}{r}$. So we can find such an $r$ if $\frac{1}{q}-\frac{1}{p}>1-\frac{2}{p}$, equivalently if $\frac{1}{p}+\frac{1}{q}>1$. Hence we may choose $r$ to ensure that
$(V^i, \frac{\partial V^i}{\partial S})$ is an $S$-controlled rough path.

Take $X_t:=(S_t,\sigma_t)$ and choose the reduced-rough-path lift ${\mathbb{X}_t}$ so that
$\mathbf{X}:=(X, \mathbb{X})$ satisfies
\[
[\mathbf{X}_t] = \left(
\begin{array}{cc}
[\mathbf{S}_t] & 0 \\
0 & 0 
\end{array}
\right)
\]
where $\mathbf{S}=(S, \mathbb{S})$ is the
reduced-rough-path lift of $S$ given by equation
\eqref{eq:bracketCondition}.

By the rough-path It\^o formula we have
\begin{equation}
\begin{split}
V^i&(S_t,\sigma_t,t)=V^i(S_0,\sigma_0,0)
+ \int_0^t \frac{\partial V^\correction{i}}{\partial t} \Big|_u d \correction{u}
\\
&+ 
\left(
\left(
\left(\begin{array}{c}
\frac{\partial V^i}{\partial S} \\
\frac{\partial V^i}{\partial \sigma}
\end{array}
\right),
\left(\begin{array}{c c}
\frac{\partial^2 V^i}{\partial S^2} & \frac{\partial^2 V^i}{\partial S \partial \sigma} 
\\
\frac{\partial^2 V^i}{\partial S \partial \sigma} &
\frac{\partial^2 V^i}{\partial \sigma^2}
\end{array}
\right)
\right)
\cdot
(X, {\mathbb X})
\right)_{0,t}
+ \int_0^t \correction{\frac{1}{2}} \frac{\partial^2 V^i}{\partial S^2} \Big|_u d [\mathbf S]_\correction{u}
\label{eqn:roughIntegral2dA}
\end{split}
\end{equation}
with the final integral being a Young integral. Using equation \eqref{eq:enhancerFromBracket} we may write down the lift ${\mathbb{X}}$.
\begin{equation*}
{\mathbb X}_{s,t}=\frac{1}{2}(S_{s,t}\odot S_{s,t}
+ 2 S_{s,t}\odot \sigma_{s,t} + \sigma_{s,t}\odot \sigma_{s,t}) + [\mathbf X]_{s,t}).
\end{equation*}
where $\odot$ denotes the symmetric tensor product.
Our hypothesis that $\frac{1}{p}+\frac{1}{q}>1$ together with Lemma \ref{lemma:productControls} allows us to write
\[
|S_{s,t}\odot \sigma_{s,t}|
\leq \| S \|_{\pvar;[s,t]} \| \, \sigma \|_{\qvar;[s,t]} \leq w(s,t)^{1+\epsilon}
\]
for some control function $w$ and $\epsilon>0$. Since we also assume $q<2$, We can bound the square of $\sigma_{s,t}$ in the same way. Hence, if we use the definition of the rough integral to write equation
\eqref{eqn:roughIntegral2dA}
as the limit of compensated Riemann sums,
all the second order derivative terms involving differentiation by $\sigma$ will vanish as the mesh tends to 0 (by equation \eqref{eqn:controlLimit}).

It therefore follows from equation \eqref{eqn:roughIntegral2dA} that
\begin{equation}
\begin{split}
V^i(S_t,\sigma_t,t)&=V^i(S_0,\sigma_0,0) + \int_0^t \frac{\partial V^\correction{i}}{\partial t}\Big|_u  d \correction{u} 
\\
&\quad + \int_0^t \frac{\partial V^i}{\partial S} \Big|_u \, d {\mathbf S}_u
+ \int_0^t \frac{ \partial V^i}{\partial S}\Big|_u  \, d \sigma_u
+ \int_0^t \correction{\frac{1}{2}} \frac{\partial^2 V^i}{\partial S^2}\Big|_u  d [\mathbf S]_\correction{u}
\label{eqn:roughIntegral2dB}
\end{split}
\end{equation}
where the second integral is a rough integral but the next two integrals are Young integrals. This is because if one writes
out the integrals in equations \eqref{eqn:roughIntegral2dA}
and \eqref{eqn:roughIntegral2dB} in terms
of their standard expressions as limits of compensated Riemann sums (equation \eqref{eq: definition of rough integral} and \cite[Exercise 6.9]{friz2010}), they
only differ by the vanishing terms.
\medskip 

Since the $V^i$ satisfy
the Black--Scholes PDE we obtain the formula
\begin{equation}
V^i_t = V^i_0 + \int_0^t \frac{\partial V^i}{\partial \sigma}\Big|_u d \sigma_u + \left((\Delta^i, \Gamma^i) \cdot (S,{\mathbb S}) \right)_{0,t}
\label{eq:representableAsIntegral2}
\end{equation}
with the first integral being a Young integral. By equation \eqref{eq:vegaProportionalToGamma} we know that
\[
\sum_{i=0}^2 q^i_t \frac{\partial V^i}{\partial \sigma} = 0
\]
at all times. Hence we will be able to apply the same argument
used to prove Theorem \ref{thm:gammaHedging}
to prove our result.
Since we have equation \eqref{eq:representableAsIntegral2} in place of equation
\eqref{eq:representableAsIntegral},
two modifications are needed. First,
when estimating the rough-path terms, we will need to use the
control function arising from the version of Gubinelli's theorem for controlled rough paths of $(p,r)$-variation regularity.

Second, we will need to bound
a Young-integral term where we previously bounded a Riemann
integral. The additional estimate needed for the Young integral terms is
    \begin{align*}
        \left\lvert \sum_{i=0}^{n} q^i_s \int_s^t \frac{\partial V^i}{\partial \sigma}\Big|_u  \, d\sigma_u \right\rvert    
&\leq \left\lvert \sum_{i=0}^{n} q^i_s  \int_s^t \frac{\partial V^{i}}{\partial \sigma} \Big|_u  \, d\sigma_u
-
\sum_{i=0}^{n} q^i_s \frac{\partial V^{i}}{\partial \sigma} 
\Big|_s \sigma_{s,t} \right\rvert  \\             
&\leq \sum_{i=0}^{n}\lvert q^i_s \rvert\left\lvert \int_s^t \frac{\partial V^{i}}{\partial \sigma} \Big|_u \, d \sigma_u - 
\frac{\partial V^{i}}{\partial \sigma} 
\Big|_s \sigma_{s,t}
\right\rvert         \\
&\leq C \sum_{i=0}^{n}\lvert q^i_s| \, \left\| \frac{\partial V}{\partial \sigma}\right\|_{\pvar;[s,t]} \, \left\| \sigma \right\|_{\qvar;[s,t]}
    \end{align*}
where the last line uses the Young--Lo\`eve estimate \cite[Theorem 6.8]{friz2010}. Applying
Lemma \ref{lemma:productControls} to the $p$- and $q$-variation controls together with the condition $\frac{1}{p}+\frac{1}{q}>1$ allows us to write
\begin{equation*}
    \left\lvert \sum_{i=0}^{n} q^i_s \int_s^t \frac{\partial V^i}{\partial \sigma}\Big|_u  \, d\sigma_u \right\rvert   
    < C w(s,t)^{1+\epsilon}
\end{equation*}
for some control $w(s,t)$ and $\epsilon>0$. Hence
we may apply the estimate  \eqref{eqn:controlLimit} to find
    \begin{equation*}
        \sum_{[s,t]\in\pi}
        \left\lvert \sum_{i=0}^{n} q^i_s \int_s^t \frac{\partial V^i}{\partial \sigma}\Big|_u  \, d\sigma_u \right\rvert   \to 0 \text{ as } \|\pi\|\to 0.
    \end{equation*}
This allows us to prove the appropriate analogue of 
Theorem \ref{thm:gammaHedging}
and hence we
obtain the desired result.
\end{proof}

We remark that a general pattern has emerged:
the same estimates used to show the convergence of integrals written as Riemann sums allow us to prove replication results.
\medskip

Our results above show that in a market with no volatility smile, we can replicate an infinite-dimensional family of payoff functions using only a finite number of derivative instruments. It is impossible to do this using a static replication strategy.

Our strategy does not require the trader to use any information about the paths $\sigma_t$ and $S_t$ beyond what is known at time $t$ and the assumed regularity. The trader's strategy computes the Greeks using the classical constant-volatility Black--Scholes pricing formulae with the instantaneous implied volatility $\sigma_t$ as input.

In this example, replication is possible even though the probability model is not known, even up to equivalence of measures. There is an infinite-dimensional family of inequivalent arbitrage-free ${\mathbb Q}$-measure models compatible with the pricing model \eqref{eq:blackScholesPricingModel}: one can use
any time-dependent volatility model with root-mean-square volatility equal to $\sigma_t$.

We also find it interesting to note that this
hedging methodology also works even if the realised signal cannot arise from any classical arbitrage-free pricing model. It is true that if one knows the signal will come from
a specific ${\mathbb P}$-measure model which contains arbitrage, one can exploit this to replicate anything. However, in reality, the ``true'' dynamics are unknown and this may prevent such an arbitrage
being exploited. Even if the dynamics were known it might be infeasible to compute an arbitrage strategy. Moreover, the notion that there really are ``true'' dynamics underpinning the market is highly debatable. The gamma-hedging strategy remains effective despite these issues.

The discussion above explains in part why we believe it is interesting to give a rough-path theory
proof of the effectiveness of our strategy rather than a probabilistic proof using semi-martingales.
Another reason is that classical arbitrage-free pricing models do not typically consider market impacts.
The gamma-hedging strategy will remain effective even in the extreme case that
an adversary can select the  path taken by the signal in full knowledge of the strategy that we pursue. It also remains robust to more realistic forms of price impact.

\medskip

Let us conclude this section with a few remarks on the economic interpretation of these results and their comparison with the Black--Scholes model.

The classical replication result of Black--Scholes shows the effectiveness of the sell-side strategy of delta hedging. In the Black--Scholes model, a trader can replicate options that buy-side investors wish to purchase at the Black--Scholes price. If the trader charges a little more than the Black--Scholes price, they will then be able to make a risk-free profit, subject to the assumption that the stock price follows that of the Black--Scholes model.

Our result is similar, except that, rather than assuming properties of the dynamics of the stock price, the trader assumes certain non-probabilistic properties of exchange-traded option prices. The validity of either assumption is an empirical question. In the case of the pre-October 1987 market,
there is good historical evidence that the volatility smile was approximately flat. Our approach seems to apply well to this historical market. On the other hand, as early as 1963, Mandelbrot reported evidence that return distributions deviated significantly from the predictions of the Black--Scholes model \cite{mandelbrot1963variation}. 

In both replication approaches, for the result to be financially useful, it is important that the replication error remains small under small perturbations of the modelling assumptions. In the case of the classical Black--Scholes approach, this is established by the Fundamental Theorem of Derivative Trading. In our approach one can appeal to the good continuity properties of  rough differential equations to obtain a similar result, see \cite{andreiThesis} for details.

Girsanov's theorem shows that the physical-measure model of Black--Scholes theory is equivalent to a martingale measure and hence the model admits no arbitrage.
In our model, the formal definition of arbitrage does not apply as our model does not assume the existence of a physical probability measure and this
features in the definition of arbitrage. Nevertheless, clearly one should require our pricing model to be consistent in some sense: for example there should be no violation of put-call parity. Our pricing model
will be consistent in this sense as there will always exist at least one classical arbitrage-free physical measure ${\mathbb P}$, namely ${\mathbb P}={\mathbb Q}$, which gives the same prices.

There are many other issues one should consider when implementing a sell-side strategy. These include the practicalities of calibrating pricing models vs.\ statistical models, the impact of transaction costs and market jumps. Our approach should be viewed as a
simple theoretical model which is compatible with the practices of calibrating to market data and gamma-hedging.

\medskip

We have considered the historical options market as an example because of its simplicity,
and because of the straightforward empirical evidence on the shape of the volatility smile
at this time. In contemporary markets one might instead assume that a more sophisticated ${\mathbb Q}$-measure model with a curved volatility smile can be used to price the chosen hedging instruments
to obtain similar results.

\end{correctionenv}

\section{Path-dependent instruments}
\label{sec:pathDependent}

\subsection{Controllable payoffs}
\label{sec:pathDependentTheory}

To apply Theorem \ref{thm:gammaHedging} to path-dependent derivatives
we will need to identify circumstances under which the risk-neutral
price of a derivative in the diffusion model \eqref{eq:diffusionSDE}
obeys the rough-path integral formula given by equation \eqref{eq:representableAsIntegral}. \correction{We will assume from this point forward that the matrix $\sigma(\hat{S},\correction{t})$ is invertible for all $\hat{S} \in U$ and all $t \in [0,T]$. We will also assume
that the drift $\mu(\hat{S},t)$ is zero.}

If $G\in L^2(\Omega)$ represents the payoff of a derivative, and $V^G_t$
is its price at time $t$, then by the \correction{m}artingale representation theorem \correction{and our assumption on the invertibility of $\sigma$}, we may define $\Delta^G$ to be the $\mathcal{L}({\mathbb R}^d,{\mathbb R})$-valued predictable process such that 
\begin{equation}
V^G_t = V^G_0 + \int_0^t \Delta_u^G \, d \correction{\tilde{S}}_u
\label{eq:delta}
\end{equation}
where the integral on the right is an It\^o integral.

We would like to consider when this It\^o integral may be written as a
rough-path integral. Let us summarize the relevant rough-path theory
(see \cite{friz2010,frizhairer2020}).

\begin{definition}
Given a local martingale $M$ we define its {\em It\^o enhancement} by the It\^o integral
$$
\mathbb{M}^{\Ito}_{s,t} = \int_s^t M_{s,u} \,\correction{\otimes}\, d M_u.
$$
\end{definition}

\begin{proposition}[It\^o and rough-path integrals coincide] \label{prop: ito and rough path integrals coincide}
Let $M$ be a local martingale with $M_t = M_0 + \int_0^t \phi_udW_u$ for each $t\in[0,T]$. Assume that, for almost all $\omega\in\Omega$, $\left(M(\omega),\phi(\omega)\right)\in\mathcal{D}^{\correction{p}\text{-var}}_{W(\omega)}$. Let $\left(Y(\omega),Y'(\omega)\right)\in\mathcal{D}^{\correction{p}\text{-var}}_{M(\omega)}$ for almost all $\omega\in\Omega$. If $Y$, $Y'$ are adapted, then almost surely, 
\begin{equation*}
    \int_0^t Y_udM_u = \left((Y,Y')\cdot (M,\mathbb{M}^{\Ito})\right)_{0,t}.
\end{equation*}
where the left integral is an It\^o integral.
\end{proposition}
The result is well known. We give a proof in Appendix \ref{proofs} for
completeness.

This motivates the following definition.

\begin{definition}[Controllable payoffs]
   $G \in L^0(\Omega;{\mathbb R})$ is said to be an \correction{$\tilde{S}$}\emph{-controllable payoff of $\correction{p}$-variation regularity} if $\Delta^G$ almost surely admits a Gubinelli derivative $\Gamma$ with respect to \correction{$\tilde{S}$ equipped with its It\^o enhancement}. That is, for almost all $\omega\in\Omega$, we have $\left(\Delta^G(\omega),\Gamma^G(\omega)\right)\in\mathcal{D}^{\correction{p}\text{-var}}_{S(\omega)}$ \correction{where $\Delta^G(\omega)(t):=\Delta^G_t(\omega)$ and
   $\Gamma^G(\omega)(t):=\Gamma^G_t(\omega)$}.
   
\correction{Suppose} in addition, $G$ is continuous in the sup norm \correction{when we regard $\Omega$ as the space of continuous paths of Brownian motion starting at $0$}. \correction{Suppose also that there is a continuous map}
\begin{align*}
         \zeta:C^{p\text{-var}}(\R^d) &\rightarrow C^{p\text{-var}}([0,T],\R^{d}) \oplus C^{p\text{-var}}([0,T],\R^{d^2}) \oplus  C^{\frac{p}{2}\text{-var}}([0,T]^2,{\mathbb R}^{d})
\end{align*}
\correction{
\noindent satisfying $\zeta(\omega)=(\Delta^G(\omega), 
        \Gamma^G(\omega), R^G(\omega))$
almost surely, where $R^G(\omega)(s,t)=\Delta^G_t-\Delta^G_s - \Gamma^G_{s,t} \tilde{S}_s$.} In these
circumstances we will say $G$ is {\em continuously controlled}.
\end{definition}

\correction{We now note the following lemma.}
\begin{lemma}
\label{lem:sControllableImpliesWControllabel}
\correction{Under our assumption} that $\sigma$ is invertible, $G$ is \correction{$\tilde{S}$}-controllable if and only
if it is $W$-controllable.
\end{lemma}
A proof is given in Appendix \ref{proofs} as Lemma \ref{lemma: being W controlled is equivalent to being diffusion controlled}.
Our proof contains an explicit formula
for the relationship between the Greeks in this situation.

\correction{As a result of this lemma, we may focus our attention on $W$-controllable payoffs henceforth.}

\medskip

If all the hedging instruments and the instrument being replicated are $W$-controllable, we know that \eqref{eq:representableAsIntegral} will hold almost surely. Hence, by Theorem \ref{thm:gammaHedging}, the error of discrete-time gamma-hedging on a grid $\pi$ will almost surely tend to 0 as $\|\pi\|\to 0$
so long as all the payoffs are $W$-controllable.

This, in itself, is a rather weak result: there are already known results
showing that the delta-hedging strategy converges almost surely for appropriately chosen sequences of grids $\pi_n$ with mesh tending to zero, albeit not for all possible choices \cite{bick1994dynamic}. However, we will now consider how to use the continuity of the rough-path integral in its integrand and integrator in the appropriate rough-path topologies to obtain a sure result for continuously controlled payoffs.

\medskip

A smooth path $X_t \in C^{\infty}\left([0,T];\mathbb{R}^d\right)$
has a canonical lift $\mathbf X=(X,{\mathbb X}^G)$ given by
\[
{\mathbb X}^G_{s,t}=\int_s^t X_{s,r} \otimes d X_r.
\]
We will call this the geometric lift of $X_t$. However this
is not the only choice of lift: we may define its Brownian lift 
$X^B=(X,{\mathbb X}^B)$ by
\[{\mathbb X}^B_{s,t}={\mathbb X}^G_{s,t}\correction{- \frac{1}{2}} (t-s)I\]
where $I$ is the identity matrix. We call this the Brownian lift
because its rough bracket $[\mathbf X]_t$ is equal to that of \correction{It\^o}-enhanced Brownian motion.

\begin{definition}
For $p \in (2,3)$, we define ${\cal B}^p([0,T]; {\mathbb R}^d)$, \correction{the space of pseudo-Brownian paths}, to be the closure of
\[
\{ (X_t, {\mathbb X}^B_t) \mid X_t \in C^{\infty}([0,T],{\mathbb R}^d), \; X_0=0 \}
\]
in the \correction{rough-path} $p$-variation norm.
\end{definition}

\correction{By \cite[Exercise 2.9]{frizhairer2020} piecewise-linear paths equipped with the Brownian lift lie in ${\cal B}_p$. By \cite[Proposition 3.6]{frizhairer2020}, dyadic piecewise-linear approximations to Brownian motion equipped with the geometric lift converge to Brownian motion with the Stratonovich lift. Hence Brownian motion with its It\^o lift almost surely lies in ${\cal B}_p$.}

The continuity properties of the rough integral will allow us to prove sure results because, as the next proposition asserts, the set of paths
of the form $(W, {\correction{\mathbb W}}^{\Ito})$ is dense in \correction{${\cal B}^p([0,T]; {\mathbb R}^d)$}. 

\begin{proposition}[Density of Brownian paths] \label{prop: paths from diffusions are dense in space of diffusive paths}
Suppose $\tilde{\mathbf{W}}_t$ in ${\cal B}^p\left([0,T];\mathbb{R}^d\right)$
then for any $\epsilon>0$, 
\begin{equation*}
    \correction{\mathbb{Q}}\left(\vertiii{\mathbf{W}^\Ito-\mathbf{\tilde{W}}}_{p\text{-var};[0,T]} < \epsilon \right) > 0 
\end{equation*}
\end{proposition}
See Appendix \ref{proofs} \correction{for a proof}.

While conditional expectations are usually only defined almost surely, we will now explain how to  define expectations conditioned on ${\cal F}_t$ surely. Given a continuous path $\tilde{W}_t$ starting at 0 and defined on an interval $[0,t_0]$,
we may define a probability space $\correction{(}\tilde{\Omega}\correction{,\tilde{{\cal F}},\tilde{\mathbb Q})}$  \correction{by taking $\tilde{\Omega}$} to be the set of continuous
paths on $[t_0, T]$ starting at $W_{t_0}$, and \correction{choosing ${\tilde{\cal F}}$ and ${\tilde{\mathbb{Q}}}$} such that if $\tilde{W} \in \tilde{\Omega}$ then $(\tilde{W}_{t\correction{+}t_0}-\tilde{W}_{t_0})$ is a Brownian motion on $[0,T-t_0]$
\correction{generating ${\tilde{\cal F}}$}. Given a measurable function $G  :\Omega \to {\mathbb R}$ we define
\begin{equation}
{\mathbb E}(G \mid (\tilde{W}_t)_{\correction{t \in}[0,t_0]} ):={\mathbb E}_{\correction{\tilde{\mathbb{Q}}}}(G(\tilde{W}_{\cdot})).
\label{eq:pathDepdendentExpectation}
\end{equation}
This is defined surely and is almost surely equal to the conditional expectation in the sense that
\[
{\mathbb E}(G \mid {\cal F}_t )(W_{\cdot} ) = {\mathbb E}(G \mid (W_t)_{\correction{t \in}[0,t_0]} ) \quad \text{a.s.}
\]

\begin{theorem}[Rough-path Clark--Ocone formula]\label{thm: sure clark ocone formula diffusive rough paths}
Suppose that $G$ is continuously $W$-controllable then for any ${\mathbf W} \in B^p\left([0,T];\mathbb{R}^d\right)$ \correction{and $t_0 \in [0,T]$}
\begin{equation*}
    \mathbb{E}(G  \mid  (\tilde{W}_t)_{\correction{t \in}[0,t_0]} ) = \mathbb{E}[G] + \left((\Delta^G(W),\Gamma^G(W))\cdot\left(W,\mathbb{W}\right)\right)_{0,t_{\correction{0}}}.
\end{equation*}
\end{theorem}
\begin{proof}
By definition of $\Delta^G$ We know that 
\begin{equation*}
    \mathbb{E}(G  \mid  {\cal F}_{t_{\correction{0}}} ) = \mathbb{E}[G] + 
    \int_0^{\correction{t_0}} \Delta^G_u \, d W_u.
\end{equation*}
for all $t_0$, and so we also have that
\begin{equation*}
    G = \mathbb{E}(G \mid {\cal F}_{t_\correction{0}}) + 
    \int_{\correction{t_0}}^T \Delta^G_u \, d W_u.
\end{equation*}
Hence
\begin{equation}
G( W_{\cdot} ) = \mathbb{E}(G \mid {(W_{\correction{t}})}_{\correction{t} \in [0,t_{\correction{0}}]})
+ \left((\Delta^G(W),\Gamma^G(W))\cdot\left(W,\mathbb{W}\right)\right)_{t_{\correction{0}},T}
\label{eq:sureCOBackward}
\end{equation}
almost surely by the equivalence of It\^o and rough integrals. The continuity of $R$ ensures that $\Gamma$ is surely a Gubinelli derivative for $\Delta$ and hence the rough integral surely exists if ${\mathbf W} \in B^p\left([0,T];\mathbb{R}^d\right)$. 
Then by the continuity of all the terms and Proposition \ref{prop: paths from diffusions are dense in space of diffusive paths}, equation \eqref{eq:sureCOBackward} holds surely
for ${\mathbf W} \in B^p\left([0,T];\mathbb{R}^d\right)$.
\correction{By equating equation \eqref{eq:sureCOBackward} with the same equation when $t_0=0$ and simplifying, the result follows.}
\end{proof}

\begin{corollary}
Suppose $W_t$ is a path starting at $0$ with finite $p$-variation, $2<p<3$.
If $G^i$ for $0 \leq i \leq n$ are continuously W-controllable payoffs,
and if a derivative with payoff $G^i$ can be purchased at each time $t_\correction{0}$ 
for the price
\[
{\mathbb E}(G^i \mid (\tilde{W}_t)_{\correction{t \in }[0,t_0]} )
\]
then the profit or loss of the discrete-time gamma-hedging strategy to replicate $G^0$ using these derivatives tends to 0 as the mesh of the partition tends to 0.
\end{corollary}
\begin{proof}
Let $C^{0,\pvar}_{0}$ denote the closure of the set of smooth paths starting at $0$ in the $p$-variation norm and let $C^{\pvar}_{0}$ denote the set of all paths of finite $p$-variation starting at 0. The result  for $W_t \in C^{0,\pvar}_{0}$ is an immediate consequence of Theorem \ref{thm:gammaHedging} and Theorem \ref{thm: sure clark ocone formula diffusive rough paths}. If $2<q<p$ then $C^{0,\qvar}_{0} \subseteq C^{\pvar}_{0}$ \correction{\cite[Corollary 5.35]{friz2010}}. This gives the
result.
\end{proof}

\medskip

In order to identify the Gubinelli derivatives of It\^o integrals
we will make frequent use of the following proposition.

\begin{proposition}[It\^o integral is regularizing]\label{prop: regularity of ito integral}
Let $W$ be a $m$-dimensional Brownian motion. Let $\phi:\Omega\times [0,T]\rightarrow \mathbb{R}^{d\times m}$ be a continuous adapted process with $\mathbb{E}\left[\int_0^T (\phi_s\phi_s^\top)_{ij}ds\right]<\infty$, for all $1\leq i,j\leq d$. Set $X_t = \int_0^t \phi_u dW_u$, 
for all $t\in [0,T]$. Suppose that almost-all sample paths of $\phi$ are of finite $p$-variation. Then, with probability one:
    \begin{enumerate}[label=(\roman*)]
        \item $X\in C^{p\text{-var}}([0,T];\mathbb{R}^{d})$;
        \item  $(X,\phi)\in\mathcal{D}^{\correction{p}\text{-var}}_{W}$, that is to say, the remainder $R^{X}$, defined as $R^X_{s,t}=X_{s,t}-\phi_s W_{s,t} = \int_s^t \phi_{s,u}dW_u$, is of finite $\correction{\frac{p}{2}}$-variation;
        \item for any $f\in \mathcal{C}^{3}_b(\mathbb{R}^d,\mathbb{R}^{d\times m})$,  $(f(X),Df(X)\phi)$ belongs to $\mathcal{D}_W^{\correction{p}\text{-var}}$. 
    \end{enumerate}
\end{proposition}
\correction{We do not claim the result is new, but as we could not find a proof of the second part in the literature, we give a proof in Appendix \ref{proofs}.}

\subsection{Examples}
\label{sec:examples}

\subsubsection{Simple payoffs}

Let $C^{\infty}_{\text{poly}}(\R^{n}; \R)$ denotes the space of smooth functions with all derivatives of polynomial growth. 
\begin{definition}[Simple payoffs]
    For each $h\in C^{\infty}\left([0,T];\mathbb{R}^d\right)$, set
    \begin{align*}
    W(h)&\coloneqq h_T W_T - h_0 W_0 - \int_{0}^{T} W_t \, dh_t \\
        &= \int h_t \, d W_t \quad a.s.
    \end{align*}    
    The set of all payoffs of the form
    \begin{equation*}
        G = f\left(W(h_1),\dots,W(h_n)\right),
    \end{equation*}
    where $h_1,\dots,h_n\in C^{\infty}$ and $f\in C^{\infty}_{\text{poly}}\left(\mathbb{R}^{n};\mathbb{R}\right)$, is the space of \emph{simple payoffs with maturity $T$}. We denote this space $\mathcal{S}^{\infty}$.
\end{definition}
We have defined $W(h)$ above as a Young integral with continuous integrand, but have noted that it is almost surely equal to an It\^o integral. The Young-integral formulation makes it clear that $G$ is defined surely on $\Omega$ and is continuous in the sup norm.

\begin{theorem} \label{theorem: simple payoffs are rough hedgeable}
    Each $G\in\mathcal{S}^{\infty}$ is continuously $W$-controllable with \begin{equation*}
    \Delta_u^G = \mathbb{E}_u[D_uG] \text{ and } \Gamma^G_u = \mathbb{E}_u[D_u\Delta^G_u],~u\in[0,T]
    \end{equation*}
    \correction{when the delta and gamma are computed in the Bachelier model $S=W$ and $D_u$ denotes the Malliavin derivative evaluated at time $u$.}
\end{theorem}

\begin{proof}
	For brevity, we denote $X=\left(W(h_1),\dots,W(h_n)\right)$. By 
    definition of the Malliavin derivative,
	\begin{equation*}
	    D_uG = \sum_{i=1}^n \frac{\partial f(X)}{\partial x_i}h_i(u).
	\end{equation*}
	By the Clark--Ocone formula
	\begin{equation*}
	  \Delta^G_u = \sum_{i=1}^n h_i(u)\mathbb{E}_u\left[\frac{\partial f(X)}{\partial x_i}\right].
   \end{equation*}
   Hence
   \begin{equation}
	    \mathbb{E}_u\left[D_u\Delta^G_u\right] = \sum_{1\leq i,j\leq n} h_i(u) \otimes h_j(u)\mathbb{E}_u\left[\frac{\partial^2 f(X)}{\partial x_i\partial x_j}\right]. \label{eq: gamma of simple payoffs}  
	\end{equation}
	Applying Clark--Ocone to each $\frac{\partial^2 f(X)}{\partial x_i\partial x_j}$, $1\leq i,j\leq n$, we have
	\begin{equation*}
	    \mathbb{E}_u\left[\frac{\partial^2 f(X)}{\partial x_i\partial x_j}\right] = \mathbb{E}\left[\frac{\partial^2 f(X)}{\partial x_i\partial x_j}\right] + \int_0^u \left(\sum_{k=1}^n h_k(s)\mathbb{E}_s\left[\frac{\partial^3 f(X)}{\partial x_i\partial x_j \partial x_k}\right)\right]dW_s,
	\end{equation*}
	for all $u\in[0,T]$. By Proposition \ref{prop: regularity of ito integral}, $\left\{\mathbb{E}_u\left[\frac{\partial^2 f(X)}{\partial x_i \partial x_j}\right]:u\in[0,T]\right\}\in C^{p\text{-var}}$ almost-surely. Also, $h_ih_j\in C^\infty$, $\forall i,j$. Again, by Proposition \ref{prop: regularity of ito integral} and \eqref{eq: gamma of simple payoffs}, we deduce that $\left\{\mathbb{E}_u\left[D_u\Delta^G_u\right]:u\in[0,T]\right\}$ belongs to $C^{p\text{-var}}$ almost surely. For the remainder term, let $(u,v)\in\Delta_T$. We have that
	\begin{align*}
	    R_{u,v} &\coloneqq \Delta_v^G - \Delta_u^G - W_{u,v}\mathbb{E}_u\left[D_u \Delta^G_u\right] \\
	    &= \sum_{i=1}^n \left\{ h_i(v)\mathbb{E}_v\left[\frac{\partial f(X)}{\partial x_i}\right]-h_i(u)\mathbb{E}_u\left[\frac{\partial f(X)}{\partial x_i} \right] \right\}
        \\     
        &\quad-\sum_{i=1}^n \sum_{j=1}^n h_i(u) \otimes h_j(u) \cdot W_{u,v} \mathbb{E}_u\left[\frac{\partial^2 f(X)}{\partial x_i \partial x_j}\right] \\
	    &= \sum_{i=1}^n h_i(u)\left\{\mathbb{E}_v\left[\frac{\partial f(X)}{\partial x_i}\right]-\mathbb{E}_u\left[\frac{\partial f(X)}{\partial x_i}\right] \right\} \\
        &\quad
     -\sum_{i=1}^n \sum_{j=1}^n h_i(u) \otimes h_j(u) \cdot W_{u,v} \mathbb{E}_u\left[\frac{\partial^2 f(X)}{\partial x_i\partial x_j}\right]
	     + \sum_{i=1}^n\left(h_i(v)-h_i(u)\right)\mathbb{E}_v\left[\frac{\partial f(X)}{\partial x_i}\right] \\
	    &= \sum_{i=1}^n h_i(u)\int_u^v \sum_{j=1}^n \left\{h_j(s)\mathbb{E}_s\left[\frac{\partial^2 f(X)}{\partial x_i\partial x_j}\right]-h_j(u)\mathbb{E}_u\left[\frac{\partial^2 f(X)}{\partial x_i\partial x_j}\right]\right\}dW_s \\
	    &\quad + \sum_{i=1}^n\left(h_i(v)-h_i(u)\right)\mathbb{E}_v\left[\frac{\partial f(X)}{\partial x_i}\right],
	\end{align*}
	where we have used the Clark--Ocone formula for $\frac{\partial f(X)}{\partial x_i}$, $1\leq i\leq n$, in the last line. By Proposition \ref{prop: regularity of ito integral}, we deduce that the process $R\in C^{\frac{p}{2}\text{-var}}$ and thus that $\Gamma^G_u = \mathbb{E}_u\left[D_u\Delta^G_u\right]$ as claimed.
\end{proof}

Note that the simple payoffs are dense in $L^2(\Omega)$.

Recall that the Malliavin derivative \cite[Chapter 1]{nualart2006malliavinCalculus} is defined by
first defining its action on simple payoffs. This gives an unbounded, but closeable, operator. One uses the $L^{\tilde{p}}$ closure of its graph to obtain
the space $\mathbb{D}^{1,\tilde{p}}$ for any $\tilde{p}\geq 1$ of Malliavin differentiable functions.  Thus by relating the delta and gamma of
simple payoffs to the Malliavin derivative, we have obtained a fundamental
connection between Malliavin derivatives and the delta and gamma
for any payoffs that can be well approximated in both Malliavin and rough-path topologies. This idea is pursued further in \cite{andreiThesis}.

Note that the Gubinelli derivative $\Gamma$ is a symmetric tensor for simple payoffs. As a result it is possible to replicate
simple payoffs by gamma-hedging strategies in European options.

\subsubsection{Asian options \correction{in the Bachelier model}}
Consider payoffs of the form $G=f(A_T)$ where $A_T = \frac{1}{T}\int_0^T W_udu$ is the arithmetic average of the stock price over the lifetime of the option and $f$ is smooth. By It\^o's lemma applied to $\{tW_t\}_{t\in[0,T]}$, we see that
\[\int_0^T W_u\, du = -\int_0^T u\, dW_u + TW_T = \int_0^T (T-u)\, dW_u,
\]
so that $A_T = \int_0^T(1-\frac{u}{T})dW_u$. Using Theorem \ref{theorem: simple payoffs are rough hedgeable} with $h(u) = 1-\frac{u}{T}$, we conclude that $G=f(A_T)$ is smoothly $W$-controllable. Moreover, the Greeks are given by $\Delta^G_u = \left(1-\frac{u}{T}\right)\mathbb{E}_u[f'(A_T)]$ and $\Gamma^G_u = \left(1-\frac{u}{T}\right)^2\mathbb{E}_u[f''(A_T)]$.

\subsubsection{Signature payoffs}

We call a row vector $\alpha = (j_1,\dots,j_l)$, where $j_i\in\{0,1\correction{,\ldots, d}\}$ for $i\in\{1,2,\dots,l\}$, a \emph{multi-index} of length $l\coloneqq l(\alpha)\in\mathbb{N}$. Denote by $n(\alpha)$ the number of components of $\alpha$ which are equal to zero and by $\mathcal{M}$ the set of all multi-indices. That is
\begin{equation*}
    \mathcal{M} \coloneqq \left\{ (j_1,\dots,j_l):j_i\in\{0,1,2,\ldots, d\}, i\in \{1,\dots,l\}, l\in\mathbb{N} \right\}\cup \{v\},
\end{equation*}
where $v$ denotes the multi-index of length zero. Given $\alpha$ and $\beta$, we write $\alpha \star \beta$ for the concatenation of the indices. For instance, $(0,1)\star (1,0,0)=(0,1,1,0,0)$. Lastly, we write $\alpha-$ for the multi-index in $\mathcal{M}$ obtained by deleting the last component of $\alpha$. For instance, $(1,0,0)- = (1,0)$.

\begin{definition}[Iterated time-augmented Brownian motion]
Write $X_t = (t,W_t)$ for time-augmented Brownian motion. For any $\alpha\in\mathcal{M}$, define the iterated integral
\begin{equation*}
    I_{\alpha;t} \coloneqq \begin{cases}
			1, & \text{if $l(\alpha)=0$,} \\
            \int_0^t I_{\alpha-;s}ds, & \text{if $l\geq 1$ and $j_l=0$,} \\
            \int_0^t I_{\alpha-;s}dW_s, & \text{if $l\geq 1$ and $j_l\correction{\neq 0}$.}
		 \end{cases}
\end{equation*}
In contrast to simple payoffs, signature payoffs are only defined almost surely on $C^0([0,T])$. One can define a signature payoff surely if the lift ${\mathbb W}_t$ is given in addition to the trace.

\end{definition}
\begin{lemma} \label{lemma: iterated integral is an ito integral}
    Suppose that $\alpha \in\mathcal{M}$. If $l(\alpha)>n(\alpha)$, then we can uniquely write $\alpha = \beta \star (i)\star (0)^k$ for some unique $k\in\{0,1,2,\dots\}$ and $\beta \in \mathcal{M}$ and $1 \leq i \leq d$. Furthermore, 
    \begin{equation} \label{eq: iterated integral is an ito integral}
        I_{\alpha;T} = \int_0^T \frac{(T-t)^k}{k!}I_{\beta;t}dW^i_t.
    \end{equation}
    This holds almost surely if $W_t$ is taken to be a Brownian motion and the integrals are It\^o integrals, and surely if $W_t$ is taken to be a rough path with a given trace ${\mathbb W}_t$ and all integrals are taken to be rough-path integrals.
\end{lemma}
\begin{proof}
    We will proceed by induction on $k$. The case $k=0$ is true by definition of the iterated integrals.  From the induction hypothesis and either the stochastic Fubini theorem \cite[p.\ 210]{protter2005stochasticintegration} or the rough-path Fubini theorem \cite{gerasimovicHairer} we may compute:
    \begin{align*}
        I_{\beta \star (i) \star (0)^{k+1}; t}&=\int_0^T I_{\beta \star (i) \star (0)^k; t} dt \\
        &=\int_0^T \int_0^t \frac{(t-s)^k}{k!} I_{\beta ; s} \, dW^i_s \, dt \\
        &=\int_0^T \int_s^T \frac{(t-s)^k}{k!} I_{\beta ; s} \, dt \, dW^i_s \\
        &=\int_0^T \frac{(T-s)^{k+1}}{(k+1)!} I_{\beta ; s} \, dW^i_s.
    \end{align*}   
\end{proof}
\begin{theorem}
    For any $\alpha\in\mathcal{M}$, the payoff $I_{\alpha;T}$ is $W$-controllable. Furthermore, we have the following formula to compute the \correction{delta and gamma in the case of the Bachelier model $S=W$}:
    \begin{equation*}
    (\Delta^{I_{\alpha;T}})^{i_1} = 
    \begin{cases*}
     0, & \text{if $l(\alpha)=n(\alpha)$}\\
     \frac{(T-t)^k}{k!}I_{\gamma\star(0);t}, & \text{if $\alpha = \gamma\star(0,i_2)\star (0)^k$}\\
    \frac{(T-t)^k}{k!}I_{\gamma\star(i_1);t}, , & \text{if $\alpha = \gamma\star (i_1,i_2) \star (0)^k$}
    \end{cases*}
    \end{equation*}
    and
    \begin{equation*}
    (\Gamma^{I_{\alpha;T}})^{i_1.i_2} = 
    \begin{cases*}
    0, & \text{if $l(\alpha)=n(\alpha)$}\\
    0  & \text{if $\alpha = \gamma\star(0,i_2)\star (0)^k$}\\
    \frac{(T-t)^k}{k!}I_{\gamma;t}, & \text{if $\alpha = \gamma\star (i_1,i_2) \star (0)^k$}
    \end{cases*}
    \end{equation*}    
    where $1 \leq i_1, i_2 \leq d$.
\end{theorem}
\begin{proof}
    Whenever $l(\alpha)=n(\alpha)$, $I_{\alpha;T}$ is a deterministic function of $T$ (or a constant) so that $\left(\Delta^{I_{\alpha;T}}, \Gamma^{I_{\alpha;T}}\right) \equiv (0,0)$. Suppose now that $l(\alpha)>n(\alpha)$ and that $\alpha = \beta\star(i_2)\star(0)^k$ for some $k\in\{0,1,2,\dots\}$,  $\beta\in\mathcal{M}$ and $1\leq i_2 \leq d$. Then by the definition of the delta and Lemma \ref{lemma: iterated integral is an ito integral}, we immediately deduce that $\Delta^{I_{\alpha;T}}_t = \frac{(T-t)^k}{k!}I_{\beta;t}$. For the gamma, if $\beta= \gamma\star(0)$, then $I_{\beta;t}$ is a Young integral and thus has vanishing Gubinelli derivative. If $\beta = \gamma\star(i_1)$, then by definition of $I_{\beta;t}$ and properties of the rough-path integral, we have that 
    \begin{equation*}
        \left(\int_0^tI_{\gamma\star(i_1);u}dW_u\right)' = I_{\gamma;t}.
    \end{equation*}
    By the Leibniz rule of Gubinelli derivatives, we conclude the result.
\end{proof}

The so-called Universal Approximation Theorem \cite[Proposition 4.5]{lyons2020} shows that any sup-norm-continuous payoff can be approximated with arbitrary accuracy in the sup norm by a finite linear combination
of these iterated integrals on any compact set of rough paths. If one has a probability model on the space of paths, one can choose the compact set to ensure that the probability of a path not lying in this set is arbitrarily small.  Such a linear combination of iterated integrals is called a linear signature payoff function. We have shown that all linear signature payoff functions are $W$-controllable.

However, this is not an entirely satisfactory result from a financial point
of view as the signature payoffs are only defined almost surely, as are their
delta and gamma. As a result signature payoffs are not continuously $W$-controllable. However, if one supposes that the iterated integrals defining the
delta and gamma are known at all times, then Lemma \ref{lemma: iterated integral is an ito integral} shows that the equation \eqref{eq:representableAsIntegral} will still hold surely.

In contrast to simple payoffs, the gamma of a signature payoff may not be
a symmetric tensor. This implies that general signature payoffs cannot be replicated using the gamma-hedging
strategy using just European options.

\subsection{Barrier options}

\begin{example} \label{example: price of a single barrier one touch option}
	The time $0$ price of the no-touch option with maturity $T$ and barrier $B>0=W_0$ \correction{in the Bachelier model with S=W} is $X^B_0 = 2\Phi\left(\frac{B}{\sqrt{T}}\right)-1$. Consequently, the price, delta and gamma of the option at some time $t\in[0,T]$ are respectively given by
	\begin{align*}
	    X^B_t &= \mathds{1}_{M_t < B}\left(2\Phi\left(\frac{B-W_t}{\sqrt{T-t}}\right)-1\right), \\
	    \Delta^B_t &= \frac{2\mathds{1}_{M_t<B}}{\sqrt{T-t}}\phi\left(\frac{B-W_t}{\sqrt{T-t}}\right), \text{ and } \\
	    \Gamma^B_t &= \frac{2(B-W_t)\mathds{1}_{M_t<B}}{(T-t)^{\frac{3}{2}}}\phi\left(\frac{B-W_t}{\sqrt{T-t}}\right).
	\end{align*}
    Here $\Phi$ is the distribution function of the standard normal, $\phi$ its density \correction{and $M_t$ is the supremum process of $W$.}
	In particular, this option is not $W$-controllable because both the delta and gamma may have discontinuous sample paths.
\end{example}

In order to obtain a controllable payoff we will need to smooth the barrier
option in some way. The benefits of smoothing payoffs are well known to practitioners: for example smoothing of the payoff function is often done to obtain better convergence of numerical schemes, even for vanilla options, see for example \cite{heston2000, wade2007}.

\begin{definition}[Mollifier] \label{definition: mollifier}
A smooth function $\rho:\mathbb{R}^n\rightarrow \mathbb{R}$ is said to be a \emph{mollifier} \correction{\cite[1.14]{giusti1984}} if 
\begin{enumerate}
    \item $\rho$ is compactly supported,
    \item \correction{$\rho$ is zero outside of a compact subset of the open unit ball,}
    \item $\int_{\mathbb{R}^n}\rho(x)dx=1$.
\end{enumerate}
\end{definition}

\begin{definition}[Drifted Brownian motion and running maximum/minimum]
Let $g:[0,T]\rightarrow \mathbb{R}$ and $W$ be a Brownian motion. We denote by $W^{-g}$ the Brownian motion $W$ drifted by the function $g$. That is $W^{-g}_t \coloneqq W_t - g(t)$, for each $t\in[0,T]$. We denote the corresponding running maximum and minimum by $M^{-g}$ and $m^{-g}$ respectively. That is $M^{-g}_t \coloneqq \sup_{s\leq t} W^{-g}_s$ and $m^{-g}_t \coloneqq \inf_{s\leq t} W^{-g}_s$, for each $t\in[0,T]$.
\end{definition}

\begin{definition}[Smoothed no-touch] \label{definition: proxy double no-touch}
Given $\epsilon_1<0<\epsilon_2$, smooth functions $g_1,g_2:[0,T]\rightarrow \mathbb{R}$ and another smooth function $f:\mathbb{R}\rightarrow [0,1]$ such that
\begin{equation*} f(x) = 
    \begin{cases}
            1, &\text{ if } g_1(T) < x < g_2(T), \\
            0, &\text{ if }  x \leq g_1(T) + \correction{\tfrac{1}{2}} \epsilon_1 \text{ or } g_2(T) + \correction{\tfrac{1}{2}} \epsilon_2\leq x,
    \end{cases}
\end{equation*}
we define a payoff function by
\begin{equation*}
    F^{\epsilon_1,\epsilon_2}(g_1,g_2;f)(W_\cdot) \coloneqq \mathds{1}_{\epsilon_1< m^{-g_1}_T, M^{-g_2}_T < \epsilon_2} f(W_T).
\end{equation*}

Given in addition a mollifier $\rho$, we define the payoff of the {\em smoothed no-touch} to be
\begin{equation*}
    G^{b_1,b_2,\rho}(\epsilon_1, \epsilon_2, g_1,g_2;f) \coloneqq \int_{\mathbb{\R}}\int_{\mathbb{\R}}
    \correction{\frac{1}{b_1 b_2}}\, 
    \rho\left(\correction{\frac{x_1}{b_1}},\correction{\frac{x_2}{b_2}}\right) F^{\epsilon_1,\epsilon_2}(g_1-x_1,g_2-x_2;f) d\correction{x}_2 d\correction{x}_1.
\end{equation*}
\end{definition}

\begin{theorem} \label{theorem: smoothed no-touch is controllable}
The smoothed no-touch payoff $G = G^{b_1,b_2,\rho}(\epsilon_1, \epsilon_2, g_1,g_2;f)$ is continuously $W$-controllable \correction{for sufficiently small positive values of $b_1$ and $b_2$}.
\end{theorem}
See Appendix \ref{proofs} \correction{for a proof.}

\subsection{Superhedging}
\label{sec:superhedging}

Simple payoffs are dense in $L^2$, and the Universal Approximation Theorem
allows to approximate derivatives using signature payoffs outside a set of low
probability. Neither result is particularly satisfactory as we would like to know which derivatives
can be approximated surely. In addition, we would like to know which options can be approximated by replication strategies using only European options or other realistic financial derivatives. 

For simplicity, we will only consider
the case of $1$-dimensional noise.

We write $C^0$ for $C^0\left([0,T];\mathbb{R}\right)$, the space of continuous $\mathbb{R}$-valued functions on $[0,T]$ which start at $0$, equipped with the supremum norm.

\begin{definition}[Approximately controllable payoff]
	An option with continuous payoff $F:C^0\rightarrow \mathbb{R}$ is said to be \emph{approximately controllable} if there exists a family of continuously $W$-controllable options with payoffs $\{G^\epsilon\}_{\epsilon>0}$ such that
	\begin{enumerate}
		\item $G^\epsilon(\omega) \geq F(\omega)$, $\forall \epsilon>0$ and $\forall\omega \in C^0$,
		\item $\lim\limits_{\epsilon\rightarrow 0} \mathbb{E}[G^\epsilon] = \mathbb{E}[F]$. 
	\end{enumerate}
	Such a family of payoffs $\{G^\epsilon\}_{\epsilon>0}$ is said to be a \emph{sequence of controllable payoffs approximating} $F$.
\end{definition}

\begin{theorem}
	Let $F:C^0\rightarrow \mathbb{R}$ be continuous. Suppose there exists a continuously $W$-controllable $G:C^0\rightarrow \mathbb{R}$ with $G(\omega)\geq F(\omega)$ for all $\omega\in C^0$ and with $\mathbb{E}^\mathds{Q}[G]<\infty$. Then $F$ is approximately controllable.
\end{theorem}

\begin{proof}
	Without loss of generality, we may assume that $G\equiv 0$. The general case then follows by applying the proof to $G-F$. Given $N\in\mathbb{N}$, define 
	\begin{equation*}
    F_N(\omega):=\correction{\frac{1}{N}} -\frac{1}{N} \sum_{n=1}^{N^2} \mathds{1}_{\omega \in F^{-1}(-\infty,-\frac{n}{N})}.
    \end{equation*}
    We may express $F_N$ in terms of the ceiling function as follows: 
    \begin{align*}
    F_N(\omega)&=\frac{1}{N} -\frac{1}{N} \sum_{n=1}^{N^2} \mathds{1}_{n< - N F(\omega)} \\
    &= \correction{ \frac{1}{N} + \frac{1}{N} \max\left\{ -\sum_{n=1}^\infty \mathds{1}_{n< - N F(\omega)}, -N^2 \right\}} \\
    &=\correction{ \frac{1}{N} + \max\left\{ - \frac{\lceil -N F(\omega) \rceil}{N}, -N \right\}.}   
    \end{align*}
    \correction{We now see that $F_N(\omega) \geq \frac{1}{N} - \frac{\lceil -N F(\omega) \rceil}{N} \geq F(\omega)$ and that $F_N(\omega)\to F(\omega)$ as $N \to \infty$.}
    By the continuity of $F$, $F^{-1}\left(-\infty,-\frac{n}{N}\right)$ is open and so $F$ is measurable. By the Dominated Convergence Theorem, $\lim_{N\rightarrow \infty} \mathbb{E}^\mathds{Q}[F_N] =  \mathbb{E}^\mathds{Q}[F]$. So for any $\epsilon >0$, we may fix $N$ large enough to ensure that
	\begin{equation*}
	\mathbb{E}^\mathds{Q}[F_N] < \mathbb{E}^\mathds{Q}[F] + \frac{\epsilon}{4}.
	\end{equation*}
	As  $F^{-1}\left(-\infty,-\frac{n}{N}\right)$ is an open set in a separable space, we may write it as a countable union of open balls in $C^0$. More precisely, let $B(h,\epsilon)$ denote the open ball around $h\in C^\infty \subset C^0$ in the supremum norm. That is
	\begin{equation*}
	    \omega \in B(h,\epsilon) \iff \lVert \omega-h\rVert_{\infty;[0,T]}<\epsilon \iff -\epsilon < m^{-h}_T(\omega),M^{-h}_T(\omega)<\epsilon.
	\end{equation*}
	As smooth functions are dense in $C^0$, we have that for each $n\in\{1,\dots,N^2\}$, there exist $\{h_{i,n}\}_{i\geq 1} \in C^\infty$ and positive $\{\epsilon_{i,n}\}_{i\geq 1}$ such that 
	\begin{equation*}
	F^{-1}\left(-\infty,-\frac{n}{N}\right) = \bigcup_{i=1}^\infty B(h_{i,n},\epsilon_{i,n}).
	\end{equation*}
	We can choose $M_n \in \mathbb{N}$ sufficiently large so that 
	\begin{equation*}
	\mathds{\correction{Q}}\left(F^{-1}\left(-\infty,-\frac{n}{N}\right)\right)-\frac{\epsilon}{4N} \leq \mathds{\correction{Q}}\left(\bigcup_{i=1}^{M_n} B(h_{i,n},\epsilon_{i,n})\right) \leq \mathds{\correction{Q}}\left(F^{-1}\left(-\infty,-\frac{n}{N}\right)\right).
	\end{equation*}
	Define 
	\begin{equation*}
	\tilde{F}\coloneqq \frac{1}{N}-\frac{1}{N}\sum_{n=1}^{N^2}\mathds{1}_{\bigcup_{i=1}^{M_n} B(h_{i,n},\epsilon_{i,n})}.
	\end{equation*}
	Then 
	\begin{equation}
	\mathbb{E}^\mathds{Q}[\tilde{F}] < \mathbb{E}^\mathds{Q}[F_N]+\frac{\epsilon}{4} < \mathbb{E}^\mathds{Q}[F]+\frac{\epsilon}{2}, \label{eq: expectation of F_N is bounded}
    \end{equation}
    \begin{equation}
	F(\omega) \leq F_N(\omega) \leq \tilde{F}(\omega),~\forall \omega\in C^0. \label{eq: F_N is bounded}
	\end{equation}
	For each $n\in\{1,\dots,N^2\}$ and $i\in\{1,\dots,M_n\}$, define
	\begin{equation*}
	F_{i,n}(\omega) \coloneqq \mathds{1}_{\{\lVert \omega-h_{i,n}\rVert_{\infty;[0,T]}>\epsilon_{i,n}\}}(\omega) 
	\end{equation*}
	For each $n\in\{1,\dots,N^2\}$, let
	\begin{equation}\label{eq: definition of tildeF_n}
	\tilde{F}_n(\omega) \coloneqq \prod_{i=1}^{M_n}F_{i,n}(\omega) = \begin{cases}
	1, &\text{if $\omega\in \bigcap_{i=1}^{M_n}B(h_{i,n},\epsilon_{i,n})^c$} \\
	0, &\text{otherwise.}
	\end{cases}
	\end{equation}
	So that $1-\tilde{F}_n(\omega) = \mathds{1}_{\bigcup_{i=1}^{M_n}B(h_{i,n},\epsilon_{i,n})}$ and by the inclusion-exclusion principle,
	\begin{equation*}
	    \mathbb{E}[1-\tilde{F}_n] = \correction{\mathbb{Q}}\left(\bigcup_{i=1}^{M_n}B(h_{i,n},\epsilon_{i,n})\right) = \sum_{k=1}^{M_n} \left((-1)^{k-1}\sum\limits_{\substack{I\subseteq \{1,\dots,M_n\} \\ \lvert I\rvert = k}}\correction{\mathbb{Q}}(B_I)\right),
	\end{equation*}
	where $B_I = \bigcap_{i\in I}B\left(h_{i,n},\epsilon_{i,n}\right)$. For any $I\subseteq \{1,\dots,M_n\}$, the payoff $\mathds{1}_{B_I}$ is a no-touch option with the time-dependent barriers
	\begin{equation*}
	        t \mapsto \pm\inf_{i\in I} \left\{\epsilon_{i,n}+h_{i,n}(t)\right\}.
	\end{equation*}
	These barriers are, in general, no longer smooth but they are continuous on the bounded interval $[0,T]$. Thus, they may be approximated by smooth functions from below/above with arbitrary accuracy. This means that we can approximate the payoffs $\tilde{F}_n$ a sequence of payoffs $\tilde{G}^{\delta}_n$ such that
	\begin{enumerate}
	    \item $\tilde{F}_n\leq \tilde{G}^{\delta}_n$, for each $\delta >0$,
	    \item $\lim\limits_{\delta\rightarrow 0}\mathbb{E}[\tilde{G}^{\delta}_n-\tilde{F}_n]=0$, and 
	    \item $\tilde{G}^{\delta}_n$ is a finite sum of no-touch options with smooth time-dependent barriers, for each $\delta>0$.
	\end{enumerate}
	By Theorem \ref{theorem: smoothed no-touch is controllable}, we can approximate each $\tilde{G}^\delta_n$ with a sequence of smoothed no-touch options, which are continuously $W$-controllable payoffs. These payoffs, $\{\tilde{G}^{\delta,\rho}_n\}_\rho$, are such that $\tilde{G}^{\delta,\rho}_n \geq \tilde{G}^\delta_n$, and $\lim\limits_\rho \mathbb{E}[\tilde{G}^{\delta,\rho}_n] = \mathbb{E}[\tilde{G}^\delta_n]$. By construction, $\tilde{G}^{\delta,\rho}_n\geq \tilde{G}^\delta_n \geq \tilde{F}_n$ and
	\begin{equation*}
	    \lim\limits_\rho \lim\limits_{\delta\rightarrow 0}\mathbb{E}[\tilde{G}^{\delta,\rho}_n-\tilde{F}_n] = \lim\limits_\rho \mathbb{E}[\tilde{G}^{\delta,\rho}_n-\tilde{G}^\delta_n] + \lim\limits_{\delta\rightarrow 0}\mathbb{E}[\tilde{G}^\delta_n-\tilde{F}_n] = 0.
	\end{equation*}
	So we can pick some smooth $\rho$ and $\delta>0$ sufficiently small so that $\mathbb{E}[\tilde{G}^{\delta,\rho}_n-\tilde{F}_n] \leq \frac{\epsilon}{2N}$, for each $n\in\{1,\dots,N^2\}$. Finally, define the $W$-controllable payoff
	\begin{equation*}
	\bar{G}^\epsilon \coloneqq \frac{1}{N}-\frac{1}{N}\sum_{n=1}^{N^2}\left(1-\tilde{G}^{\delta,\rho}_n\right).
	\end{equation*}
	Then,
	\begin{align*}
	\mathbb{E}[\bar{G}^\epsilon] &= \frac{1}{N}-\frac{1}{N}\sum_{n=1}^{N^2}\mathbb{E}[1-\tilde{G}^{\delta,\rho}_n] \\
	&\leq \frac{1}{N}-\frac{1}{N}\sum_{i=1}^{N^2}\left\{\mathbb{E}[1-\tilde{F}_n]-\mathbb{E}[\tilde{G}^{\delta,\rho}_n-\tilde{F}_n]\right\} \\
	&\leq \frac{1}{N}-\frac{1}{N}\sum_{n=1}^{N^2}\left(\correction{\mathbb{Q}}\left(\cup_{i=1}^{M_n}B(h_{i,n},\epsilon_{i,n})\right)-\frac{\epsilon}{2N}\right) \\
	&= \left(\frac{1}{N}-\frac{1}{N}\sum_{n=1}^{N^2}\correction{\mathbb{Q}}\left(\cup_{i=1}^{M_n}B(h_{i,n},\epsilon_{i,n})\right)\right)+\frac{\epsilon}{2} \\
	&= \mathbb{E}[\tilde{F}]+\frac{\epsilon}{2} < \mathbb{E}[F] + \epsilon.
	\end{align*}
	We also see that
	\begin{align*}
	\bar{G}^\epsilon &= \frac{1}{N}-\frac{1}{N}\sum_{n=1}^{N^2}\left((1-\tilde{F}_n)+(\tilde{F}_n-\tilde{G}^{\delta,\rho}_n)\right)
	\\
	&\geq \frac{1}{N}-\frac{1}{N}\sum_{i=1}^{N^2}\mathds{1}_{\bigcup_{i=1}^{M_n}B(h_{i,n},\epsilon_{i,n})} = \tilde{F} \geq F,
	\end{align*}
	using (\ref{eq: F_N is bounded}). It follows that for each $\epsilon>0$, there exists a sequence of controllable payoffs via $\left\{\bar{G}^\epsilon \right\}_{\epsilon>0}$ which approximate $F$.
\end{proof}

To extend our result to higher-dimensional noise, the additional ingredient needed is to prove that higher-dimensional no-touch options with smooth barrier manifolds can be superhedged using $W$-controllable payoffs. This requires a higher-dimensional analogue
of Theorem \ref{theorem: smoothed no-touch is controllable}, which can be proved in the same way using PDE theory. In this case the gamma of the smoothed option can be related to the partial derivatives of the solution to the PDE in a formula analogous to equation \eqref{eq:gammaAsPartialDerivatives} in the appendix. Hence in higher dimensions one can superhedge continuous payoffs $F$ with $W$-controllable hedging instruments that have symmetric gamma tensors such as European options. By contrast, signature payoffs representing L\'evy-area terms need not be continuous and can have non-symmetric gamma.

\section{Conclusion}

We have shown how the gamma-hedging strategy may provide some insight into the benefits of of the market practices of calibrating
pricing models to market prices and gamma hedging.

The central ingredient in our approach is to develop rough-path versions of the Clark--Ocone formula and we have given a number of examples. The thesis \cite{andreiThesis}, studies the relationship between Malliavin calculus and the Gubinelli derivative, defining a closeable operator for whose domain these notions coincide and hence for which a Clark--Ocone formula holds. The same thesis also explores the relationship between the Gubinelli derivative and functional It\^o calculus.

\bibliographystyle{plain} % We choose the "plain" reference style
\bibliography{references.bib}

\appendix

\section{Proofs}
\label{proofs}

\begin{proof}[Proof of \ref{thm: gubinelli theorem} (Gubinelli's Theorem)]
\begin{correctionenv}The result is proved on the case of H\"older norms in Theorem 4.10 of \cite{frizhairer2020}. We show only the changes needed for $p$-variation results. Suppose $s\leq u \leq t$. Set
\[
\Xi_{s.t}=Y_sX_{s,t}+Y'_s\mathbb{X}_{s,t}
\]
then we have
\[
\delta \Xi_{s,u,t}:=\Xi_{s.t} - \Xi_{s.u} - \Xi_{u.t} = -R_{s,u}^Y X_{u,t} - Y^\prime_{s,u} \mathbb{X}_{u,t}.
\]
Since $p<3$, $x\to x^{\frac{p}{3}}$ is concave, so
\begin{align*}
|\delta \Xi_{s,u,t}|^\frac{p}{3}
&\leq  \| R^Y \|^\frac{p}{3}_{\halfpvar;[s,u]}
\| X \|^\frac{p}{3}_{\pvar;[u,t]}
+ 
\| Y^\prime \|^\frac{p}{3}_{\pvar;[s,u]}
\| \mathbb{X} \|^\frac{p}{3}_{\halfpvar;[u,t]} \\
&= w_{R^Y,\frac{p}{2}}(s,t)^\frac{2}{3}
 w_{X,p}(s,t)^\frac{1}{3}
 + w_{Y^\prime,p}(s,t)^\frac{1}{3}
   w_{\mathbb{X},\frac{p}{2}}(s,t)^\frac{2}{3}.
\end{align*}
Applying Lemma \ref{lemma:productControls}  we see that the right-hand side is a control.
The existence of the limit and the bound now follow from Proposition 3.5 of \cite{armstrong2020option}. The argument of
Theorem 4.10 \cite{frizhairer2020} establishes the continuity of the integral map.
\end{correctionenv}
\end{proof}

\begin{correctionenv}
\begin{lemma}
\label{lemma:convexity}
Let $f:(0,\infty) \to \R$ be a convex, non-linear payoff function, then  $V^f(S,\sigma,\tau):=\mathrm{BS}^f(S,\sigma,\tau)$ satisfies 
$\frac{\partial^2 V^f}{\partial S^2}>0$ for all $\tau>0$.
\end{lemma}
\begin{proof}
Using the explicit formula for the pricing kernel in the Black--Scholes model, equation \eqref{eq:priceAsIntegral} yields
\[
V^f(S,\sigma,\tau) = \int_0^\infty f(u) \frac{1}{u \sigma \sqrt{2 \pi \tau}}
\exp\left( 
- \frac{(\log(u) - \log(S) + \frac{1}{2}\sigma^2 \tau)^2}{2 \sigma^2 \tau}
\right) \, du.
\]
Making the substitution $v=\frac{u}{s}$ we may write
\begin{equation}
V^f(S,\sigma,\tau) = \int_0^\infty f( S v) \frac{1}{v \sigma \sqrt{2 \pi \tau}}
\exp\left( 
- \frac{(\log(v) + \frac{1}{2}\sigma^2 \tau)^2}{2 \sigma^2 \tau}
\right) \, dv.
\label{eq:asConvexCombination}
\end{equation}
Thus, as a function of $S$, $V^f(S,\sigma,\tau)$ is a convex combination of the convex functions $S\to f(Sv)$, so must itself be a convex. Since $V^f$ is
smooth for $\tau>0$, we conclude $\frac{\partial^2 V^f}{\partial S^2}\geq 0$.

If $f$ were in addition smooth and had a positive second derivative at some point $x$, then by considering the integral \eqref{eq:asConvexCombination} in a neighbourhood of $v=\frac{x}{S}$ we could conclude that $\frac{\partial^2 V^f}{\partial S^2}$ is strictly positive for any $\tau>0$. On the other hand, if $f^{\prime\prime}\equiv0$ then $\frac{\partial^2 V^f}{\partial S^2}\equiv0$.

Since we can think of a derivative with payoff $T$ in the Black-Scholes model as a being a derivative with maturity $T-\epsilon$ and payoff $\mathrm{BS}^f(S,\sigma,\tau-\epsilon)$, we can combine the observations of the last two paragraphs to see that either $\frac{\partial^2 V^f}{\partial S^2}$ is identically zero for $\tau>\epsilon$ or is strictly positive for $\tau>\epsilon$. Hence either $\frac{\partial^2 V^f}{\partial S^2}\equiv0$ for $\tau>0$, or $\frac{\partial^2 V^f}{\partial S^2}>0$ for all $\tau>0$.
Suppose for contradiction the former occurs. It then follows from the Black--Scholes PDE that $\frac{\partial V^f}{\partial \tau}=0$. Since $V$ is linear in $S$ for $\tau>0$, by the continuity of $V$, $f$ must be linear too, a contradiction. Hence $\frac{\partial^2 V^f}{\partial S^2}>0$ for all $\tau<0$.
\end{proof}
\end{correctionenv}

\begin{proof}[Proof of Proposition \ref{prop: ito and rough path integrals coincide} (It\^o and rough-path integrals coincide)]
\ \linebreak
Using Proposition 7.1 of \cite{frizhairer2020}, the right hand side rough-path integral is equal to 
\begin{equation*}
    \left((Y,Y')\cdot (M,\mathbb{M})\right)_{0,t} = \lim_{\lVert \pi\rVert \rightarrow 0}\sum_{[u,v]\in\pi} \left\{Y_uM_{u,v}+Y'_u\phi'_u\mathbb{W}^{\Ito}_{u,v}\right\}.
\end{equation*}
For any continuous adapted process $Y$, the It\^o integral against a local martingale $M$ is defined as the limit in probability,
\begin{equation*}
    \int_0^t Y_udM_u \coloneqq \plim_{n\rightarrow \infty}\sum_{[u,v]\in\pi_n}Y_uM_{u,v},
\end{equation*}
where $\{\pi_n\}_{n\in\mathbb{N}}$ is any sequence of finite partitions with $\lim_{n\rightarrow \infty}\lVert \pi_n\rVert= 0$. By switching to a subsequence, we can assume that the convergence holds almost surely. Following the proof of Proposition 5.1 of \cite{frizhairer2020}, it now suffices to show that the difference between the It\^o and the rough-path integral, $\sum_{[u,v]\in\pi_n}Y'_u\phi'_u\mathbb{W}^{\Ito}_{u,v}$, vanishes on a set of full measure as $n\rightarrow \infty$. It will suffice to show that
\begin{equation*}
\left\lVert\sum_{[u,v]\in\pi}Y'_u\phi'_u\mathbb{W}^{\Ito}_{u,v}\right\rVert_{L^2}^2 = \mathcal{O}\left(\lVert \pi\rVert\right),
\end{equation*}
as convergence in $L^2$ implies convergence in probability, and convergence in probability implies almost-sure convergence on a subsequence. For a given finite partition $\pi=\{0=t_0<t_1<\dots<t_n=T\}$, define the \correction{piecewise-constant} martingale $N$ with $N_0\coloneqq 0$, \correction{which is constant on the intervals $[t_k,t_{k+1})$ and with increments} $N_{\correction{t_{k+1}}}-N_\correction{t_{k}}=  Y'_{t_k}\phi'_{t_k}\mathbb{W}^{\Ito}_{t_k,t_{k+1}}$. Define also the stopping times $\tau_K\coloneqq \correction{\inf}\{t\in\correction{[0,T]}:\lVert Y'_t\phi'_t\rVert_{\max}\correction{\geq} K\}$, for each $K\in\mathbb{N}$ and where $\lVert \cdot \rVert_{\max}$ denotes the max norm for matrices (picks the element with the largest absolute value). \correction{The stopped process} $N^{\tau_K}$ is also a martingale and, for each $K\in\mathbb{N}$, we have
\begin{align*}
\left\lVert\sum_{[u,v]\in\pi}Y'_{u\wedge\tau_K}\phi'_{u\wedge\tau_K}\mathbb{W}^{\Ito}_{u,v}\right\rVert^2_{L^2} &= \left\lVert\sum_{k=0}^{n-1}(N^{\tau_K}_{t_{k+1}}-N^{\tau_K}_{t_k})\right\rVert^2_{L^2} \\
    &= \sum_{k=0}^{n-1}\left\lVert N^{\tau_K}_{t_{k+1}}-N^{\tau_K}_{t_k}\right\rVert^2_{L^2} \\
    &\leq K^2 \sum_{k=0}^{n-1}\left\lVert \mathbb{W}^{\Ito}_{t_k,t_{k+1}}\right\rVert^2_{L^2} \\
    & = \mathcal{O}\left(\lVert \pi\rVert\right),
\end{align*}
where we have used that $\left\lVert \mathbb{W}^{\Ito}_{t_k,t_{k+1}}\right\rVert^2_{L^2}$ is proportional to $\lvert t_{k+1}-t_k\rvert^2$ in the last line. As this is true for all $K\in\mathbb{N}$ and $\tau_K\rightarrow\infty$ as $K\rightarrow \infty$, this concludes the proof.
\end{proof}

To prove the equivalence of the notions of $\tilde{S}$- and $W$-controllable we will use a sequence of lemmas.

\begin{lemma} \label{lemma: discounted controllable option prices can be written as rough path integrals}
\correction{
If $\sigma(\tilde{S_t},t)$ invertible for all $\tilde{S_t}$, $t \in [0,T]$ and $\mu=0$ then $G$ is $\tilde{S}$-controllable if and only if there exist adapted processes $Y$, $Y'$ such that for any $0\leq s<t\leq T$,
}
\begin{equation*}
    \mathbb{E}_t^\mathbb{Q}[G]-\mathbb{E}_s^\mathbb{Q}[G] = \left((Y,Y')\cdot(\tilde{S},\tilde{\mathbb{S}})\right)_{s,t}, 
\end{equation*}
\correction{
almost-surely. Moreover, $Y$ and $Y^\prime$ are unique up to indistinguishability of processes and $Y=\Delta^G$.}
\end{lemma}
\begin{proof}
\begin{correctionenv}
Suppose $G$ is $\tilde{S}$-controllable. Equation \eqref{eq:delta} establishes the It\^o integral formula
\begin{equation*}
\mathbb{E}_t^\mathbb{Q}[G]-\mathbb{E}_s^\mathbb{Q}[G] = \int_{s,t} \Delta_u \, d\tilde{S}_u.
\end{equation*}
By our assumption on $G$, $\Delta$ almost surely admits a Gubinelli derivative $\Gamma$. We may now apply
Proposition \ref{prop: ito and rough path integrals coincide} to obtain the only-if part of our thesis.

Suppose that a suitable $(Y,Y^\prime)$ exists. By definition of $\Delta^G$ in equation \eqref{eq:delta} and the uniqueness in the martingale representation theorem, $Y$ must equal $\Delta^G$. Our result then amounts to showing uniqueness of the Gubinelli derivative up to indistinguishability of processes. According to \cite[Proposition 6.4]{frizhairer2020} if a rough path has the property of being ``truly rough'' then the Gubinelli derivative will be unique. According to \cite[Theorem 6.6]{frizhairer2020}, Brownian motion is truly rough for the values of $p$ we are considering. Using  Proposition \ref{prop: ito and rough path integrals coincide} and the fact that $W_t$ is a solution to the stochastic differential equation
\[ d W_t = \sigma^{-1}(\correction{\tilde{S}}_t, t) \, d \correction{\tilde{S}}_t, \qquad W_0=0, \]
we see that the Gubinelli derivative with respect to $\tilde{S}$ must also be unique.
\end{correctionenv}
\end{proof}

\begin{lemma}\label{lemma: W is S controllable}
    \correction{
    Assume $\sigma(\tilde{S_t},t)$ is invertible for all $\tilde{S_t}$
    and $t \in [0,T]$ and that $\mu=0$.}
    Let $\correction{\tilde{S}}(\omega)$ denote a path from the diffusion (\ref{eq:diffusionSDE}) such that $\correction{\tilde{S}}\in C^{p\text{-var}}\left([0,T];\mathbb{R}^d\right)$ and the underlying Brownian path $W\in C^{p\text{-var}}\left([0,T];\mathbb{R}^m\right)$. Then $(Z,Z')\in\mathcal{D}_W^{\correction{p}\text{-var}} \Rightarrow \left(Z,Z'\left(\sigma(\correction{\tilde{S},t})\right)^{-1}\right)\in\mathcal{D}_{\correction{\tilde{S}}}^{\correction{p}\text{-var}}$.
\end{lemma}
\begin{proof}
By Proposition \ref{prop: regularity of ito integral}, $\left(\correction{\tilde{S}},\sigma(\correction{\tilde{S},t})\right)\in\mathcal{D}_W^{\correction{p}\text{-var}}\left([0,T];\mathbb{R}^d\right)$. By the hypothesis, we know that the remainder $R^S$, given by $R^{S}_{s,t}\coloneqq \correction{\tilde{S}}_{s,t}-\sigma(\correction{\tilde{S}_s,s})W_{s,t}$, is of finite $\correction{\frac{p}{2}}$-variation. Let $R^{W}_{s,t}\coloneqq W_{s,t}-\left(\sigma(\correction{\tilde{S}_s,s})\right)^{-1}\correction{\tilde{S}}_{s,t}$. Thus,
\begin{equation*}
    R^{S}_{s,t} = \correction{\tilde{S}}_{s,t}-\sigma(\correction{\tilde{S}_s,s})\left(R^W_{s,t}+\left(\sigma(\correction{\tilde{S}_s, s})\right)^{-1}\correction{\tilde{S}}_{s,t}\right) = - \sigma(\correction{\tilde{S}_s, s})R^W_{s,t}.
\end{equation*}
Hence, $R^W$ also defines a remainder of finite $\correction{\frac{p}{2}}$-variation so that $\left(W,\left(\sigma(\correction{\tilde{S}, t})\right)^{-1}\right)\in \mathcal{D}^{\correction{p}\text{-var}}_{\correction{\tilde{S}}}\left([0,T];\mathbb{R}^{m}\right)$. By the change of signal formula given in \cite[Proposition 7.1]{frizhairer2020} it follows that $(Z,Z'\left(\sigma(\correction{\tilde{S}})\right)^{-1})\in\mathcal{D}_{\correction{\tilde{S}}}^{\correction{p}\text{-var}}$.
\end{proof}

\begin{lemma} \label{lemma: being W controlled is equivalent to being diffusion controlled}
    Let $\{\correction{\tilde{S}}_t\}_{t\in[0,T]}$ be a diffusion of the form (\ref{eq:diffusionSDE}) \correction{with invertible $\sigma$
    and $\mu=0$}. Then it holds almost surely that
    \begin{multline*}
\left(Y,Y'\right)\in\mathcal{D}^{\correction{p}\text{-var}}_{\correction{\tilde{S}}}\left([0,T];\mathbb{R}^{d}\right) \Longrightarrow \\
     \left(Y^\top\sigma(\correction{\tilde{S},t}), Y^\top D\sigma(\correction{\tilde{S}})\sigma(\correction{\tilde{S},t})+\sigma(\correction{\tilde{S},t})(Y')^\top\right)\in\mathcal{D}^{\correction{p}\text{-var}}_W\left([0,T];\mathbb{R}^{ m}\right).
    \end{multline*}
    Conversely,
    \begin{multline*} \left(Y,Y'\right)\in\mathcal{D}^{\correction{p}\text{-var}}_W\left([0,T];\mathbb{R}^{m}\right)\Longrightarrow \\ \left(Y^\top A,Y^\top A'+A(Y')^\top\right)\in\mathcal{D}^{\correction{p}\text{-var}}_{\correction{\tilde{S}}}\left([0,T];\mathbb{R}^{d}\right),
    \end{multline*}
    where $A$ denotes the inverse of $\sigma(\correction{\tilde{S},t})$ and $A'$ its Gubinelli derivative with respect to $\correction{\tilde{S}}$.
    Hence $G$ is $\correction{\tilde{S}}$-controllable $\iff G$  is $W$-controllable.
\end{lemma}
\begin{proof}
For the first part, we know from Proposition \ref{prop: regularity of ito integral} that $\left(\correction{\tilde{S}},\sigma(\correction{\tilde{S},t})\right)$ belongs to $\mathcal{D}^{\correction{p}\text{-var}}_{W}\left([0,T];\mathbb{R}^{d}\right)$ so that $\left(Y^\top,(Y')^\top\sigma(\correction{\tilde{S},t})\right)\in \mathcal{D}^{\correction{p}\text{-var}}_{W}\left([0,T];\mathbb{R}^{d}\right)$ by the chain rule for Gubinelli derivatives. Similarly, 
\[
\left(\sigma(\correction{\tilde{S},t}),D\sigma(\correction{\tilde{S},t})\right)\in\mathcal{D}^{\correction{p}\text{-var}}_{\correction{\tilde{S}}}\left([0,T];\mathbb{R}^{d\times m}\right)
\]
and so $\left(\sigma(\correction{\tilde{S}},t),\left(D\sigma(\correction{\tilde{S},t})\right)\sigma(\correction{\tilde{S},t})\right)$ belongs to $\mathcal{D}^{\correction{p}\text{-var}}_W\left([0,T];\mathbb{R}^{d\times m}\right)$. Applying the Leibniz rule to the path $Y^\top \sigma(\correction{\tilde{S},t})$ yields the result. The second part is proved analogously by using \correction{Lemma} \ref{lemma: W is S controllable}. For the final part, we just apply the first two results taking $Y\equiv \Delta^G$ and $Y_t=\mathbb{E}_t[D_tG]$ respectively.
\end{proof}

\begin{proof}[Proof of Proposition \ref{prop: paths from diffusions are dense in space of diffusive paths} (Density of Brownian paths)]
\ \linebreak
    The Stratonvich enhancement of Brownian motion is defined to be
    \[
    {\mathbb W}^{\Strat}_{s,t}=\int_s^t W_{s,u} \circ d W_u,
    \]
    where the last integral is a Stratonvich integral. 
    It is proved in Theorem 5 of \cite{ledoux2002} that for $p\in(2,3)$ 
        \begin{equation*}        \correction{\mathbb{Q}}\left(\vertiii{\mathbf{W}^{\Strat}}_{p\text{-var};[0,1]}<\epsilon \right)  > 0,~\forall \epsilon >0.
        \end{equation*}        

    Let ${\mathbf 0}^B=(0,{\mathbb O})$ denote the (non-trivial) Brownian lift of the zero path.
    Using the fact that
    \begin{align*}
        \mathbb{W}^\Ito_{s,t} - {\mathbb O}^B_{s,t} =({\mathbb W}^{\Strat}_{s,t}),
    \end{align*}
    and the scaling property of the rough-path metric it follows that
    \begin{equation*}
    \correction{\mathbb{Q}}\left(\vertiii{\mathbf{W}^\Ito - \mathbf{0}^\Ito}_{p\text{-var};[0,T]} < \epsilon\right)  > 0. 
    \end{equation*}

    \correction{Since ${\cal B}^p\left([0,T];\mathbb{R}^d\right)$ is defined as the closure of a space of paths with smooth trace, by the triangle inequality we} may assume without loss of generality that ${\tilde{\mathbf{W}}}$ has a smooth trace \correction{and its Brownian lift}. Consider the rough differential equation
    \begin{equation}
    \label{eq:driftedBrownianMotion}
    dY_\correction{t} = \frac{d \tilde{W}_t}{dt} dt + d{\mathbf W}^B_t, \quad Y_0=0.
    \end{equation}    
    The solution of a rough differential equation is a controlled rough path. The Gubinelli derivative of the solution is required to match the corresponding coefficient in the differential equation. The trace of the solution is determined by requiring that the equation holds as a rough-integral equation. 
    Hence we can write down the solution of this equation in the case that the driver is
    ${\mathbf W}={\mathbf 0}^B$: it is the $0$-controlled rough path $(\tilde{W}_t, I)$ where $I$ is the constant map with value given by the identity matrix.     

    Given a rough path $(X, \mathbb X)$ and an $X$-controlled rough path $(Y,Y^\prime)$ we can canonically enhance $Y$ to obtain a new
    rough path $(Y, {\mathbb Y})$ via the Riemann sum:
    \[
    \correction{{\mathbb Y}}_{s,t}=\lim_{\|\pi\|\to0} (\sum_{\pi} Y_u \otimes Y_{u,v} + Y^\prime \otimes Y^\prime {\mathbb X}_{u,v}).
    \]
    Thus the canonical enhancement of $(\tilde{W}_t, I)$ is ${\tilde W}^B_t$. \correction{
    Write ${\mathbf{Y}}$ for the random rough-path solution of
    \eqref{eq:driftedBrownianMotion} when the driver is $(W,{\mathbb W}^\Ito)$. }
    Given $\epsilon>0$, by the continuity of the solution operator for RDEs, and
    the continuity of canonical enhancement, we can find $\delta>0$ such that 
    \begin{equation}
    \label{eq:continuityOfRDE}
        \vertiii{\mathbf{W}^\Ito(\omega)- \mathbf{0}^B }_{p\text{-var};[0,T]} < \delta \Rightarrow \vertiii{\mathbf{Y}-\tilde{\mathbf{W}}}_{p\text{-var};[0,T]} < \epsilon.
    \end{equation}    
    \begin{correctionenv}
    Define a process $X_t$ by the It\^o integral    
    \[
    X_t = -\int_0^t
\frac{d \tilde{W}_t}{dt} d W_t.
    \]
    Using Girsanov's theorem and Proposition \ref{prop: ito and rough path integrals coincide}
    we may define an equivalent probability measure $\tilde{\mathbb Q}$ by the condition
    \[
    \frac{d \tilde{\mathbb{Q}}}{d \mathbb{Q}}\Big|_{{\cal F}_t} = \exp\left( X_t - \frac{1}{2} [X]_t \right)
    \]
    such that ${\mathbf Y}$ is the It\^o ehancement of a $\tilde{{\mathbb Q}}$-measure Brownian motion.
    
    We have already shown that
    \begin{equation*}        \mathbb{Q}\left(\vertiii{\mathbf{W}^\Ito-\mathbf{0^{\text{B}}}}_{p\text{-var};[0,T]} < \delta \right) > 0. 
    \end{equation*}
    Hence by the equivalence of measures
    \begin{equation*}        \tilde{\mathbb{Q}}\left(\vertiii{\mathbf{W}^\Ito-\mathbf{0^{\text{B}}}}_{p\text{-var};[0,T]} < \delta \right) > 0. 
    \end{equation*}
    It now follows from equation \eqref{eq:continuityOfRDE} that
    \begin{equation*}        \tilde{\mathbb{Q}}\left(\vertiii{\mathbf{Y}-\mathbf{\tilde{W}}}_{p\text{-var};[0,T]} < \epsilon \right) > 0. 
    \end{equation*}
    Since $\mathbf{Y}$ is an It\^o-enhanced $\tilde{\mathbb{Q}}$-measure Brownian motion
    and $\mathbf{W}^\Ito$ is an It\^o-enhanced
    $\mathbb{Q}$-measure Brownian motion, this is equivalent to the statement
    \begin{equation*}        \mathbb{Q}\left(\vertiii{\mathbf{W}^\Ito-\mathbf{\tilde{W}}}_{p\text{-var};[0,T]} < \epsilon \right) > 0. 
    \end{equation*}
    \end{correctionenv}
\end{proof}

To prove Proposition \ref{prop: regularity of ito integral} we recall the following result.

\begin{proposition}[Continuous semimartingales enhanced to rough paths \cite{coutin2005}, Proposition 1] \label{prop: semimartingales can be enhanced to rough paths}
    Let $X=(X^1,\dots,X^d)$ be a $\mathbb{R}^d$-valued continuous semimartingale and set 
    \begin{equation*}
        \mathbb{X}^{\Ito}_{s,t} \coloneqq \int_{s}^{t}X_{s,r} \otimes dX_r \eqqcolon \left(\int_s^t X^i_{s,r}dX^j_r\right)_{1\leq i,j\leq d},
    \end{equation*}
    where the integral is in the It\^o sense. Then for any $p\in(2,3)$ and almost all paths $\omega\in\Omega$, $\boldsymbol{X}(\omega)=\left(X(\omega),\mathbb{X}^{\Ito}(\omega)\right)\in \mathcal{C}^{p\text{-var}}([0,T];\mathbb{R}^d)$.
\end{proposition}

\begin{proof}[Proof of Proposition \ref{prop: regularity of ito integral} (the It\^o integral is regularizing)]
It\^o integrals are special cases of continuous semimartingales so the first statement is a direct application of Proposition \ref{prop: semimartingales can be enhanced to rough paths}. The third part is a consequence of the chain rule for Gubinelli derivatives \cite[Theorem 7.6]{frizhairer2020}. For the second part, we first prove the result in the one-dimensional case $d=m=1$. From the Dubins--Schwarz theorem, we see that, for each $(s,t)\in\Delta_T$, there exists a Brownian motion $B$, possibly depending on $s$ and $t$, such that $R^X_t = R^X_s + B_{\int_s^t \phi_{s,u}^2du}$. \label{correction: not sure of the argument here}For any $\alpha\in(\frac{1}{3},\frac{1}{2})$, $B$ is $\alpha$-H\"older continuous with probability one. Pick an $\omega$ where this is the case. As $B(\omega)\in C^{\alpha\text{-H\"ol}}$, we have
\begin{equation*}
    \lvert R^X_{s,t}(\omega)\rvert = \left\lvert B_{\int_s^t\phi_{s,u}^2du}(\omega)\right\rvert \lesssim \left\lvert \int_s^t \phi_{s,u}^2du \right\rvert^{\alpha}.
\end{equation*}
We note that even though $B$ possibly depended on $(s,t)$, this dependence has now disappeared. For the purpose of easing the notation of this proof, we can assume that $\phi(\omega) \in C^{p\text{-var}}$. This is because $\phi$ has almost-all paths of finite $p$-variation so the subset of $\Omega$ on which both $B$ and $\phi$ are of the correct regularities is a set of full measure. As $\phi(\omega)$ is of finite $p$-variation, there exists a control function $w:\Delta_T\rightarrow [0,\infty)$ such that $\phi_{s,u}\leq w(s,u)^{\frac{1}{p}}$ for all $(s,u)\in\Delta_T$. Thus
\begin{align*}    
     \lvert R^X_{s,t}(\omega)\rvert^{\correction{\frac{p}{2}}}&\lesssim \left\lvert \int_s^t w(s,u)^{\frac{2}{p}}du \right\rvert^{\correction{\frac{p\alpha}{2}}} \\
     &\leq \left(w(s,t)^{\frac{2}{p}}(t-s)\right)^{\correction{\frac{p\alpha}{2}}} \\
     &= \correction{w(s,t)^{\correction{\alpha}}(t-s)^{\frac{p\alpha}{2}}} \\
     &\correction{\eqqcolon \bar{w}(s,t).}
\end{align*}
\correction{Using Lemma \ref{lemma:productControls}, We see that $\bar{w}(\cdot,\cdot)$ is a control function}. Consequently $R^X(\omega)\in C^{\correction{\frac{p}{2}}\text{-var}}$ for almost-all $\omega\in\Omega$. Now consider the multi-dimensional case. For each $1\leq i\leq d$ and $1\leq j\leq m$, let $Y_{ij}$ denote the one-dimensional process defined by $(Y_{ij})_{s,t}\coloneqq \int_s^t (\phi_{ij})_{s,u}dW^j_u$ so that the $i^{\text{th}}$ component of $R^X$ is given by $R^X_i = \sum_{j=1}^m Y_{ij}$. Repeating our one-dimensional proof shows that $Y_{ij}\in C^{\correction{\frac{p}{2}}\text{-var}}([0,T];\mathbb{R})$, for all $1\leq i\leq d$ and $1\leq j\leq m$. By $\correction{\frac{p}{2}}$-variation norm estimates of the form $\lVert \sum_{j=1}^m f_j \rVert_{\correction{\frac{p}{2}}\text{-var}}\lesssim \sum_{j=1}^m \lVert f_j\rVert_{\correction{\frac{p}{2}}\text{-var}}$, we deduce that $R^X_i\in C^{\correction{\frac{p}{2}}\text{-var}}([0,T];\mathbb{R})$, $\forall i$. By the equivalence of norms in finite-dimensional spaces, we can once again use the former estimate to show that $R^X = (R^X_1,\dots, R^X_d)\in C^{\correction{\frac{p}{2}}\text{-var}}([0,T];\mathbb{R}^d)$.
\end{proof}

We prove two Lemmas that will allow us to prove the controllability
of barrier options.

\begin{lemma} \label{lemma: barrier options are controllable if enough derivatives exist}
Consider an option \correction{in the Bachelier model $S=W$} whose payoff at maturity depends on $T$, the current level of the stock, $W_T$, and the running maximum of the stock, $M_T = \max_{s\in[0,T]}W_s$. That is
\begin{equation*}
    F=F(T,W_T,M_T).
\end{equation*}
Suppose that $F\in\mathds{D}^{1,r}$, then the delta exists and can be written as $\Delta^F_t = \mathbb{E}_t[D_t F]$, for all $t\in[0,T]$. $(W_t,M_t)_{t\in[0,T]}$ is Markov so we may write the price process as 
\begin{equation*}
    \mathbb{E}_t[F]\big\vert_{W_t=x,M_t=b} = \varphi(t,x,b),
\end{equation*}
for some function $\varphi:[0,T]\times \mathbb{R}\times [0,\infty)\rightarrow \mathbb{R}$. If $\varphi$ is continuously differentiable in its second argument, then
\begin{equation*}
    \Delta^F_t = \varphi_x(t,W_t,M_t),
\end{equation*}
for all $t\in[0,T]$. Furthermore, if $\varphi_{xxx}$, $\varphi_{xxt}$ and $\varphi_{xxb}$ exist and are continuous, then $F$ is a $W$-controllable payoff with
\begin{equation*}
    \Delta^F_t = \varphi_x(t,W_t,M_t) \text{ and } \Gamma^F_t = \varphi_{xx}(t,W_t,M_t),~\forall t\in[0,T].
\end{equation*}
\end{lemma}

\begin{proof}
For a given $\omega\in\Omega$ and fixed $(u,v)\in\Delta_T$, we define the following functions from $[0,1]\rightarrow \mathbb{R}$,
\begin{align*}
    \hat{\varphi}(\theta) \coloneqq \varphi\big(u+\theta(v-u), &W_u+\theta W_{u,v}, M_u+\theta M_{u,v}\big)(\omega), \\
    \hat{\varphi}_x(\theta) \coloneqq \varphi_x\big(u+\theta(v-u), &W_u+\theta W_{u,v}, M_u+\theta M_{u,v}\big)(\omega), \\
    & \vdots \\
    \hat{\varphi}_{txx}(\theta) \coloneqq \varphi_{txx}\big(u+\theta(v-u), &W_u+\theta W_{u,v}, M_u+\theta M_{u,v}\big)(\omega).
\end{align*}
Henceforth, we implicitly omit the `$\omega$' to keep the notation lighter. By Taylor's theorem, there exists a $\theta\in[0,1]$ such that $\hat{\varphi}_{xx}(1) - \hat{\varphi}_{xx}(0) = \hat{\varphi}_{xx}'(\theta)$. Hence,
\begin{align*}
\Gamma^F_{u,v} &= \varphi_{xx}(v,W_v,M_v) - \varphi_{xx}(u,W_u,M_u) \\
&= \hat{\varphi}_{xx}(1)- \hat{\varphi}_{xx}(0) \\
&= \hat{\varphi}_{xx}'(\theta) \\
&= \hat{\varphi}_{txx}(\theta)(v-u) + \hat{\varphi}_{xxx}(\theta)W_{u,v} + \hat{\varphi}_{bxx}(\theta)M_{u,v},
\end{align*}
where we have used the chain rule in the last line. For any $p\in(2,3)$, the function $x\mapsto x^p$ is convex so that we can use Jensen's inequality to see that
\begin{align*}
    \left\lvert \Gamma^F_{u,v} \right\rvert^p \leq 3^{p-1}C^p\left((v-u)^p+\lvert W_{u,v}\rvert^p +\lvert M_{u,v}\rvert ^p\right),
\end{align*}
where $C=\sup\left\{\varphi_{txx},\varphi_{xxx},\varphi_{bxx}\right\}$ is just a constant for each path. For almost all paths, $t\mapsto W_t$ and $t\mapsto M_t$ are in $C^{p\text{-var}}$ so that there exists a control function $w(\cdot,\cdot)$ with $\lVert \Gamma^F_{u,v}\rVert^p\leq w(u,v)^p$. Note that $w(\cdot,\cdot)$ will depend on each path. Hence $\Gamma^F \in C^{p\text{-var}}$ almost-surely. For the remainder, we have
\begin{align*}
R^F_{u,v} &= \Delta^F_{u,v} - \Gamma^F_uW_{u,v} \\
&= \varphi_x(v,W_v,M_v)-\varphi_x(u,W_u,M_u) -\varphi_{xx}(u,W_u,M_u)W_{u,v}\\
&= \hat{\varphi}_x(1)-\hat{\varphi}_x(0) -  \hat{\varphi}_{xx}(0)W_{u,v}, \\
&= \hat{\varphi}_x'(\theta) - \hat{\varphi}_{xx}(0)W_{u,v} \intertext{using that $\hat{\varphi}_x(1)-\hat{\varphi}_x(0)= \hat{\varphi}'_x(\theta)$, for some $\theta\in[0,1]$ by Taylor's theorem,}
&= \hat{\varphi}_{tx}(\theta)(v-u) + \hat{\varphi}_{bx}(\theta)M_{u,v} + \left(\hat{\varphi}_{xx}(\theta)-\hat{\varphi}_{xx}(0)\right)W_{u,v},  \intertext{using the chain rule on $\hat{\varphi}_x(\theta)$,}
&= \hat{\varphi}_{tx}(\theta)(v-u) + \hat{\varphi}_{bx}(\theta)M_{u,v} + \theta\hat{\varphi}_{xx}'(\tilde{\theta}) W_{u,v} \intertext{using that $\hat{\varphi}_{xx}(\theta)-\hat{\varphi}_{xx}(0)=\theta \hat{\varphi}_{xx}'(\tilde{\theta})$, for some $\tilde{\theta}\in[0,\theta]$ by Taylor's again,}
&= \hat{\varphi}_{tx}(\theta)(v-u) + \hat{\varphi}_{bx}(\theta)M_{u,v} \\
&\qquad + \theta\left\{\hat{\varphi}_{txx}(\tilde{\theta})(v-u)W_{u,v}+\hat{\varphi}_{xxx}(\tilde{\theta})(W_{u,v})^2 +\hat{\varphi}_{bxx}(\tilde{\theta})W_{u,v}M_{u,v}\right\},
\end{align*}
where we have used the chain rule on $\hat{\varphi}_{xx}(\tilde{\theta})$ in the last line. For any $p\in(2,3)$, the function $x\mapsto x^{\frac{p}{2}}$ is convex. By Jensen's inequality, this implies
\begin{multline*}
    \left\lvert R^F_{u,v} \right\rvert^{\frac{p}{2}} \leq5^{p-1}\tilde{C}^{\frac{p}{2}} \times \\ 
    \left\{(v-u)^{\frac{p}{2}}+\lvert M_{u,v}\rvert ^{\frac{p}{2}}+(v-u)^{\frac{p}{2}}\lvert W_{u,v}\rvert^{\frac{p}{2}}+\lvert W_{u,v}\rvert^p+\lvert W_{u,v}\rvert^{\frac{p}{2}}\lvert M_{u,v}\rvert^{\frac{p}{2}}\right\},
\end{multline*}
where $\tilde{C} = \sup\left\{\varphi_{tx},\varphi_{bx},C\right\}$ is just a constant for each path. By the Cauchy--Schwarz inequality, we have
\begin{align*}
    2(v-u)^{\frac{p}{2}}\lvert W_{u,v}\rvert^{\frac{p}{2}} &\leq (v-u)^p+\lvert W_{u,v}\rvert^p \intertext{and}
    2\lvert W_{u,v}\rvert^{\frac{p}{2}}\lvert M_{u,v}\rvert^{\frac{p}{2}} &\leq \lvert W_{u,v}\rvert^p+\lvert M_{u,v}\rvert^p,
\end{align*}
so that there exists some control function $\tilde{w}(\cdot,\cdot)$ with $\lvert R^F_{u,v}\rvert^{\frac{p}{2}} \lesssim \tilde{w}(u,v)^{\frac{p}{2}}$. This implies that $R^F\in C^{\frac{p}{2}\text{-var}}$ almost-surely. Hence $F$ is $W$-controllable.
\end{proof}

\begin{lemma} \label{lemma: payoff is controllable if enough derivatives exist}
Consider an option whose price at time $t$ depends on the values, also at time $t$, of a set of features given by the $d+2$-dimensional stochastic process 
\begin{equation*}
    \left\{(t,W_t,Y^1_t,Y^2_t,\dots,Y^d_t):t\in[0,T]\right\}.
\end{equation*}
Then the time $t$ price of the option when $W_t = x$ and $\left\{Y^i_t = y_i\right\}_{i=1,\dots,d}$ can be expressed as some function $\varphi(t,x,y_1,\dots,y_d)$. If
\begin{enumerate}
    \item $t\mapsto Y^i_t$ are almost-surely in $C^{p\text{-var}}$ for each $i=1,\dots,d$, and
    \item $\varphi_{xxt},\varphi_{xxx},\varphi_{xxy_1},\dots\varphi_{xxy_d}$ all exist and are continuous,
\end{enumerate}
then the payoff of the option is $W$-controllable.
\end{lemma}
The proof is a straightforward adaptation of the proof of Lemma \ref{lemma: barrier options are controllable if enough derivatives exist}.

% Long proof

\begin{proof}[Proof of Theorem \ref{theorem: smoothed no-touch is controllable} (smoothed no-touch is continuously controllable)]
\ \linebreak
For each pair $(t,y)$ in the domain $\left\{(t,y):t\in[0,T],  g_1(t) + \epsilon_1 \leq y \leq g_2(t) +\epsilon_2 \right\}$ and for each $i=1,2$ \correction{and $t\in [0,T]$, we make a smooth choice of diffeomorphism of $\R$, $h_t(y)$}, which sends $g_i(t) + \epsilon_i$ \correction{to $0$ when $i=1$, and $1$ when $i=2$, and which has derivative $1$ in a neighbourhood of the barriers.} 

One can check that $F^{\epsilon_1,\epsilon_2}(g_1,g_2;f)$ has not knocked out by time $T$ if and only if $h_t(W_t)\in(0,1)$ for each $t\in[0,T]$. Thus, the payoff of $\correction{F^{\epsilon_1,\epsilon_2}(g_1,g_2; f)}$ can alternatively be written as
\begin{equation*}
    F^{\epsilon_1,\epsilon_2}(g_1,g_2;f) =
    \begin{cases}
            f(W_T), &\text{ if } h_\correction{t}(W_t) \in (\correction{0},\correction{1}) \text{ for each } t\in[0,T], \\
            0, &\text{ otherwise.}
    \end{cases}
\end{equation*}
Define the process $Y = \left(Y_t = \correction{h_t}(W_t)\right)_{t\in[0,T]}$ and denote the running maximum and minimum of $Y$ as $\underline{Y}_t = \inf_{s\leq t} Y_s$ and $\overline{Y}_t = \sup_{s\leq t} Y_s$ respectively so that $F^{\epsilon_1,\epsilon_2}(g_1,g_2;f)$ can be written as the payoff of an option on the path of $Y$ instead of the path of $W$. We have
\begin{equation*}
    F^{\epsilon_1,\epsilon_2}(g_1,g_2;f) = \left(f\circ h^{-1}_T\right)(T,Y_T)\cdot \mathds{1}_{\correction{0}< \underline{Y}_T,\overline{Y}_T < \correction{1}}.
\end{equation*}
By It\^o's formula, we know that $Y$ is a diffusion with smooth coefficients and driven by $W$.  \correction{Writing $v(t,y)$} for the price at time $t$ of the option with payoff $F^{\epsilon_1,\epsilon_2}(g_1,g_2;f)$ when $Y_t = y$ and assuming the option has not knocked out by that time. \correction{We see that $v(t,y)$} satisfies \correction{a} parabolic PDE with
smooth coefficients on $[0,T)\times (b_1,b_2)$
with boundary conditions
\begin{equation} \label{eq: boundary conditions for flat barrier double barrier}
\begin{cases}
    v(t,y) &= 0, ~ t\in[0,T), \correction{y\in\{0,1\}},\\
    v(T,y) &= \left(f\circ h_T^{-1}\right)(y), ~ y \in(\correction{0,1})
\end{cases}    
\end{equation}
By Theorem 6.5.2 of \cite{lorenzi2021}, $v(t,y)$ is smooth in $t$ and $y$ on the rectangle $[0,T]\times[0,1]$. 

\correction{The solution $v(t,y)=v(t,y)$ is also a smooth function of the coefficients of the PDE and of the final boundary condition, which in turn depend smoothly on the choice of the family of diffeomorphism $h_t$. To see why the solution depends smoothly on the coefficients, we refer to arguments outlined in \cite[Chapter 3]{henry2005} which shows how to apply the Implicit Function Theorem in Banach spaces together with standard results on existence, uniqueness and well-posedness of parabolic equations to prove such results.}

    By Lemma \ref{lemma: being W controlled is equivalent to being diffusion controlled}, we know that a payoff $G$ is $Y$-controllable if and only if it is $W$-controllable. We use this property to show that the smoothed no-touch payoff is $W$-controllable.

    The price at time $t$ of an option with payoff $F^{\epsilon_1,\epsilon_2}(g_1,g_2;f)$ is
    \begin{equation*}
        \mathbb{E}_t[F^{\epsilon_1,\epsilon_2}(g_1,g_2;f)] = \mathds{1}_{\correction{0}<\underline{Y}_t,\overline{Y}_t<\correction{1}}v_{\correction{h}}(t,Y_t)
    \end{equation*}
    where we write $v_{h}$ to emphasize that the solution of the PDE depends upon the choice of diffeomorphism $h$.
    
    Hence, the price at time $t$ of $G$ when $Y_t = y$, $\underline{Y}_t = \underline{m}$ and $\overline{Y}_t = \overline{m}$ is given by the function
    \begin{equation*}
        \varphi(t,y,\underline{m},\overline{m}) = \int_{{\mathbb R}}\int_{{\mathbb R}} \frac{1}{b_1}\frac{1}{b_2} \rho\left(\correction{\frac{x_1}{b_1}},\correction{\frac{x_2}{b_2}}\right) \mathds{1}_{x_1<\underline{m},\overline{m}<\correction{1}+x_2}v_{\correction{h(x_1,x_2)}}(t,y)dx_2 dx_1.
    \end{equation*} 
    \correction{In this expression we are writing $h(x_1,x_2)$ to emphasize that the choice of diffeomorphism $h$ will depend upon $x_1$ and $x_2$. Note that it is when writing the expression above that we have used our assumption on the derivative of $h_t$ near the barriers. This ensures that the magnitude of small perturbations of the barrier is unchanged by $h$. It is also at this point that we use the fact that $b_1$ and $b_2$ are assumed to be sufficiently small.}

    \correction{Since $v_y$ depends smoothly on $h$, and $h$ can be taken to vary smoothly with $x_1$ and $x_2$, we may therefore write the price of $G$ as}
    \begin{equation*}
        \varphi(t,y,\underline{m},\overline{m}) = \int_{{\mathbb R}}\int_{{\mathbb R}} \rho\left(\correction{x_1},\correction{x_2}\right) \mathds{1}_{\correction{x_1}<\underline{m},\overline{m}<\correction{1+x_2}}u(t,y,x_1,x_2) dx_2 dx_1
    \end{equation*}
    \correction{For some smooth function $u$.}
    
    By Lemma \ref{lemma: payoff is controllable if enough derivatives exist}, it suffices to show that $\varphi_{yyy}$, $\varphi_{yyt}$, $\varphi_{yy\underline{m}}$ and $\varphi_{yy\overline{m}}$ exist and are continuous. The gamma is given by
    \begin{equation}
        \varphi_{yy}(t,y,\underline{m},\overline{m}) = \int_{\correction{\R}}\int_{\correction{\R}} \rho(\correction{x_1,x_2})\mathds{1}_{\correction{x_1}<\underline{m},\overline{m}<\correction{1+x_2}}u_{yy}(t,y,x_1,x_2)\correction{dx_2 dx_1},
        \label{eq:gammaAsPartialDerivatives}
    \end{equation}
    and thus the partial derivatives $\varphi_{yyy}$ and $\varphi_{yyt}$ are given by 
    \begin{align*}
        \varphi_{yyy} &=
\int_{\correction{\R}}\int_{\correction{\R}} \rho(\correction{x_1,x_2})\mathds{1}_{\correction{x_1}<\underline{m},\overline{m}<\correction{1+x_2}}u_{yyy}(t,y,x_1,x_2)\correction{dx_2 dx_1},
        \intertext{ and }
        \varphi_{yyt} &=
\int_{\correction{\R}}\int_{\correction{\R}} \rho(\correction{x_1,x_2})\mathds{1}_{\correction{x_1}<\underline{m},\overline{m}<\correction{1+x_2}}u_{yyt}(t,y,x_1,x_2)\correction{dx_2 dx_1},
    \end{align*}
    which are continuous by Theorem 6.5.2 of \cite{lorenzi2021}. For the remaining partial derivatives, we use that $u$ is smooth in $x_i$, $i=1,2$ and the fact that the solution of the PDE depends smoothly on its coefficients as discussed above. Using integration by partsshows that
    \begin{align*}
        \varphi_{yy\underline{m}}(t,x,\underline{m},\overline{m})
        &= \int_{{\correction{{\mathbb R}}}}\mathds{1}_{\overline{m}<\correction{x_2}}\int_{\correction{{\mathbb R}}}\mathds{1}_{x_1<\underline{m}}\frac{\partial\left(\rho(\correction{x_1,x_2})\correction{u}_{yy}(t,y,\correction{x_1,x_2})\right)}{\partial x_1}\correction{dx_1 dx_2} \\
        &\qquad - \int_{\correction{\mathbb{R}}}\mathds{1}_{\overline{m}<\correction{x_2}} \left[ \mathds{1}_{\correction{x_1}< \underline{m}}\cdot \rho(\correction{x_1,x_2})\correction{u}_{yy}(t,y,\correction{x_1,x_2}) \right]_{\correction{x_1} = -\infty}^{\correction{x_1} = +\infty} \correction{dx_2} \\
        &= \int_{{\correction{{\mathbb R}}}}\mathds{1}_{\overline{m}<\correction{x_2}}\int_{\correction{{\mathbb R}}}\mathds{1}_{x_1<\underline{m}}\frac{\partial\left(\rho(\correction{x_1,x_2})\correction{u}_{yy}(t,y,\correction{x_1,x_2})\right)}{\partial x_1}\correction{dx_1 dx_2},
    \end{align*}
    where the last term in the penultimate line vanishes because $\rho$ has compact support. As $\rho$ and $\correction{u}_{yy}$ are smooth in their parameters, we deduce that $\varphi_{yy\underline{m}}$ is continuous. A similar calculation shows that $\varphi_{yy\overline{m}}$ exists and is continuous, as required.

    The controllability of the payoff follows by applying Lemma \ref{lemma: payoff is controllable if enough derivatives exist}. To show that it is
    continuously controllable we note that
\begin{equation*}
    \lVert (S,M)-(\hat{S},\hat{M}) \rVert_{p\text{-var};[0,T]}
    \lesssim \lVert S-\hat{S}\rVert_{p\text{-var}},
\end{equation*}
since the paths $t\mapsto M_t$ are at least as regular as those of $t\mapsto S_t$. 
It is now straightforward to obtain Lipschitz bounds on the $\Delta$, $\Gamma$ and remainder terms using formulae of the type obtained in \ref{lemma: payoff is controllable if enough derivatives exist}.
        
\end{proof}

\end{document}